\documentclass[onecolumn,11pt]{article}

\newcommand{\nc}{\newcommand}
\nc{\rnc}{\renewcommand}

\usepackage{etoolbox}
\newtoggle{twocol}
\nc{\twocol}[1]{\iftoggle{twocol}{#1}{}}
\nc{\onecol}[1]{\iftoggle{twocol}{}{#1}}
\nc{\twoone}[2]{\iftoggle{twocol}{#1}{#2}}
\togglefalse{twocol}

\newcounter{MYtempeqncnt}

\usepackage{amsmath, amsthm, amssymb,color,fullpage,url,booktabs}  
\usepackage{bm}
\usepackage[colorlinks,pdftex]{hyperref} 

\newcommand{\ket}[1]{\left|#1\right\rangle}
\newcommand{\proj}[1]{\left|#1\right\rangle\left\langle #1\right|}

\newcommand{\seteq}{\mathrel{\mathop:}=}

\DeclareMathOperator{\conv}{conv}
\DeclareMathOperator{\id}{id}

\DeclareMathOperator{\tr}{tr}

\DeclareMathOperator{\Sep}{Sep}
\DeclareMathOperator{\famSep}{\mathbf{Sep}}
\DeclareMathOperator{\famSEP}{\mathbf{SEP}}

\DeclareMathOperator{\supp}{supp}
\DeclareMathOperator{\swap}{SWAP}

\DeclareMathOperator{\SEP}{SEP}
\DeclareMathOperator{\PPT}{PPT}

\def\be#1\ee{\begin{equation}#1\end{equation}}
\def\bea#1\eea{\begin{eqnarray}#1\end{eqnarray}}
\def\beas#1\eeas{\begin{eqnarray*}#1\end{eqnarray*}}
\def\ba#1\ea{\begin{align}#1\end{align}}
\def\bas#1\eas{\begin{align*}#1\end{align*}}
\def\non{\nonumber}
\def\nn{\nonumber}
\def\eq#1{(\ref{eq:#1})}

\def\ra{\rightarrow}
\def\ot{\otimes}

\newtheorem{thm}{Theorem}
\newtheorem*{thm*}{Theorem}

\newtheorem{lemma}[thm]{Lemma}

\newtheorem{lem}[thm]{Lemma}

\newtheorem{dfn}{Definition}
\newtheorem{proto}{Protocol}
\newtheorem{con}[thm]{Conjecture}
\newtheorem{rem}[thm]{Remark}

\makeatletter
\newtheorem*{rep@theorem}{\rep@title}
\newcommand{\newreptheorem}[2]{%
\newenvironment{rep#1}[1]{%
 \def\rep@title{#2 \ref{##1} (restatement)}%
 \begin{rep@theorem}}%
 {\end{rep@theorem}}}
\makeatother

\newreptheorem{thm}{Theorem}
\newreptheorem{lem}{Lemma}

\newcommand{\1}{\mathbf{1}}

\def\cA{\mathcal{A}}

\def\cD{\mathcal{D}}

\def\cL{{\cal L}}
\def\cM{{\cal M}}

\def\cP{\mathcal{P}}

\def\bbC{\mathbb{C}}
\DeclareMathOperator*{\E}{\mathbb{E}}

\def\bbN{\mathbb{N}}
\def\bbR{\mathbb{R}}

\def\benum{\begin{enumerate}}
\def\eenum{\end{enumerate}}
\def\bit{\begin{itemize}}
\def\eit{\end{itemize}}

\newcommand{\secref}[1]{Section~\ref{sec:#1}}
\newcommand{\appref}[1]{Appendix~\ref{sec:#1}}
\newcommand{\lemref}[1]{Lemma~\ref{lem:#1}}
\newcommand{\thmref}[1]{Theorem~\ref{thm:#1}}

\newcommand{\defref}[1]{Definition~\ref{def:#1}}

\newcommand*{\defeq}{\seteq}

\newcommand{\kl}[2]{D(#1\,\|\, #2)}
\newcommand{\dmid}{\,\|\,}
\nc{\hyp}{{\cal E}}
\newcommand{\adv}{\hyp_{\mathrm{adv}}}
\newcommand{\advg}{\gamma_{\mathrm{adv}}}
\newcommand{\remove}[1]{}
\DeclareMathOperator*{\pr}{\mathbb{P}}

\newcommand{\e}{\varepsilon}
\newcommand{\eps}{\varepsilon}

\def\famM{\mathbf{M}}
\def\famR{\mathbf{R}}
\def\famS{\mathbf{S}}
\def\famrho{\boldsymbol{\rho}}
\def\famsigma{\boldsymbol{\sigma}}

\nc{\LO}{\mathsf{LO}}
\nc{\ONELOCC}{\mathsf{1\text{-}LOCC}}
\nc{\LOCC}{{\mathsf{LOCC}}}
\nc{\ALL}{{\mathsf{ALL}}}
\nc{\famALL}{\mathbf{\mathsf{ALL}}}

\nc{\ERM}{D_{\famM}}
\nc{\ERMS}{{D_{\famM,\famS}}}
\nc{\ERMSinf}{{D^\infty_{\famM,\famS}}}
\nc{\ERLO}{{D_{\mathbf{LO}}}}
\nc{\ERONELOCC}{{D_{\mathbf{1\text{-}LOCC}}}}
\nc{\ERLOCC}{{D_{\mathbf{LOCC}}}}
\nc{\ERSEP}{{D_{\mathbf{SEP}}}}
\nc{\ERPPT}{{D_{\mathbf{PPT}}}}
\nc{\ERLOCCinfty}{{D^{\infty}_{\mathbf{LOCC}}}}

\begin{document}
\title{Adversarial hypothesis testing and a quantum Stein's Lemma for restricted measurements}
\author{Fernando G. S. L. Brand\~ao\thanks{Caltech},
 Aram W. Harrow\thanks{MIT},
 James R. Lee\thanks{University of Washington}, and
 Yuval Peres
\footnote{Part of this work appeared in the  Proceedings of the 5th
  conference on Innovations in theoretical computer science (ITCS '14).}}
\maketitle

\begin{abstract}
  Recall the classical hypothesis testing setting with two sets of probability
  distributions $P$ and $Q$.  One receives either $n$ i.i.d.  samples from a distribution
  $p \in P$ or from a distribution $q \in Q$ and wants to decide from which set the points
  were sampled.  It is known that the optimal exponential rate at which errors decrease
  can be achieved by a simple maximum-likelihood ratio test which does not depend on $p$
  or $q$, but only on the sets $P$ and $Q$.

  We consider an adaptive generalization of this model where the choice of $p \in P$ and
  $q \in Q$ can change in each sample in some way that depends arbitrarily on the previous
  samples.  In other words, in the $k^{\text{th}}$ round, an adversary, having observed
  all the previous samples in rounds $1,\ldots,k-1$, chooses $p_k \in P$ and $q_k \in Q$,
  with the goal of confusing the hypothesis test.  We prove that even in this case, the
  optimal exponential error rate can be achieved by a simple maximum-likelihood test that
  depends only on $P$ and $Q$.

  We then show that the adversarial model has applications in hypothesis testing for {\em
    quantum states} using restricted measurements.  For example, it can be used to study
  the problem of distinguishing entangled states from the set of all separable states
  using only measurements that can be implemented with local operations and classical
  communication (LOCC).  The basic idea is that in our setup, the deleterious effects of
  entanglement can be simulated by an adaptive classical adversary.

  We prove a quantum Stein's Lemma in this setting: In many circumstances, the optimal
  hypothesis testing rate is equal to an appropriate notion of quantum relative entropy
  between two states.  In particular, our arguments yield an alternate proof of Li and
  Winter's recent strengthening of strong subadditivity for von Neumann entropy.
\end{abstract}

\section{Introduction}\label{sec:intro}

\twoone{\IEEEPARstart{A}{symmetric}}{A symmetric} hypothesis testing is the problem of distinguishing between two sources where one wants to minimize the rate of false positives (type-1 error) subject to a constraint
on the rate of false negatives (type-2 error).
 In the case of $n$ i.i.d. samples from a classical or quantum source,
 a central result is the Chernoff-Stein Lemma~\cite{Chernoff52, HP91,
   ON00} which states that for any constant bound on the type-2 error,
 the optimal type-1 error decreases at an exponential rate whose
 exponent is given by the classical (respectively, quantum) relative
 entropy.  Similar results hold even when we generalize the problem so
 that the sources are described by an unknown parameter and one needs
 to design a test that works for any choice of the parameter~\cite{Hoeffding65}.

\medskip
\noindent
{\bf First main result: Adversarial hypothesis testing.}
In the first part of this paper (\secref{adaptive}), we generalize
this problem further to allow the parameter to vary adaptively from
sample to sample.  Since we will allow the parameter to depend
arbitrarily on previous samples, this can be thought of as {\em
  adversarial} hypothesis testing.  That is, we wish to devise a test
that can distinguish between samples from two different sets even
against an adversary who can choose the distribution in each round
based on which samples have previously been observed. 

There are some simple cases where it is not hard to see that this additional power cannot help the adversary.  For example, suppose  we are given a coin with heads probability $p$ and wish to distinguish between the cases where $p\in [0,1/3]$ and where $p\in [2/3,1]$.  It is straightforward to show that this general problem is no harder than simply distinguishing a $1/3$-biased coin from a $2/3$-biased coin; equivalently, the adversary gains no advantage from the ability to be adaptive.  On the other hand, distinguishing
between the two settings $p \in \{1/3, 2/3\}$ and $p=1/2$ is clearly impossible, as the adversary can simply choose with probability $1/2$ to flip the $1/3$-biased coin,
and with probability $1/2$ to flip the $2/3$-biased coin.  The resulting distribution of samples is indistinguishable from the one arising from $p=1/2$.
This stresses the role of {\em convexity} since even a non-adaptive adversary can simulate a convex combination of distributions by choosing randomly among them.

We will prove in \thmref{adaptive} that this property is sufficient to
characterize the optimal error rate for asymmetric hypothesis testing
against an adaptive adversary.  Specifically, if the two sources vary
over convex sets of probability distributions, then the problem is no
harder than in the i.i.d. case.  Our \thmref{chernoff} also
establishes a version of this claim for symmetric hypothesis testing.
These two results can be thought of as adversarial versions of the
classic Chernoff-Stein Lemma and Chernoff's Theorem, respectively.
Results in this direction were previously established for arbitrarily
varying sources~\cite{AVS-hypo} which can be viewed as a special case
of a non-adaptive adversary.

\medskip
\noindent
{\bf Quantum hypothesis testing, entanglement, and additivity.}
One of our main applications for our adversarial Chernoff-Stein Lemma is in quantum hypothesis testing, when the states to be distinguished need not be i.i.d.  Indeed, a recurrent challenge in quantum information theory is that even apparently i.i.d. problems can involve complicated entangled states (meaning that they cannot be written as a convex combination of independent states).  For example, the quantum capacity of an i.i.d channel requires maximizing over all $n$-component inputs, and in general it is known that achieving the capacity requires using states that are entangled across channel uses~\cite{DSS98, Hastings-additivity}.  This phenomenon in quantum information theory---where information-theoretic quantities for $n$ copies of a system are not simply $n$ times the one-copy quantity---is known generally as the ``additivity'' problem.

A similar additivity problem arises in quantum hypothesis testing when we wish to distinguish many copies of a fixed state against a family of states that include non-i.i.d. states.  One  important example is the {\em relative entropy of entanglement} $E_R$, which is a method of quantifying the entanglement in a state $\rho$ as the minimum of its relative entropy with respect to any separable (i.e. non-entangled) state.  Here, $\rho$ is a multipartite state (e.g., shared between systems $A,B,C$) and separability refers to this partition.  However, to establish the asymptotic hypothesis testing rate of $\rho$ against separable states, we need to compare $n$ copies of $\rho$ against states that are separable with respect to our original partition, but not necessarily across the different copies.  In our example, $\rho^{\ot n}$ lives on systems $A_1,B_1,C_1,\ldots,A_n,B_n,C_n$ and we need to compare against states that are separable across the $A_1\ldots A_n : B_1\ldots,B_n : C_1\ldots C_n$ partition, but possibly entangled within the $A_1,\ldots,A_n$ systems (and the $B_1,\ldots,B_n$ and $C_1,\ldots,C_n$ systems).  Indeed, such entanglement across copies is known to be necessary to compute the relative entropy of entanglement, since examples exist~\cite{VW01} where $E_R(\rho\ot \rho) < 2 E_R(\rho)$.

\medskip
\noindent
{\bf Second main result: Restricted measurements.}
A further difficulty arises in the quantum setting when we consider restricted families of measurements, such as those arising from locality restrictions.  Here, too, the optimal measurement can be entangled across copies.  Moreover, since the hypothesis testing problem involves maximizing distinguishability over allowable measurements and minimizing over states, it is possible for entanglement to either increase or decrease the rate.

One particularly relevant example for our work
involves distinguishing many copies of a state $\rho$ against a general separable state, using measurements from a class (such as 1-LOCC, defined below) which preserves the set of separable states.  This distinguishability scenario was studied extensively in \cite{MWW09,Piani09, BCY11, LW12, BH-local}.  Though it may initially seem to be an obscure question, it has found applications to understanding the quantum conditional mutual information~\cite{BCY11}, to channel coding~\cite{MW12}, and to classical algorithms for separability testing~\cite{BCY-stoc} and the small-set expansion problem~\cite{BHKSZ12}.

The main result of \secref{quantum} provides quantum versions of the Chernoff-Stein Lemma and Chernoff's theorem for restricted measurements.  The main idea is that the deleterious effects of entanglement in this setting are no worse than what could be achieved by an adaptive adversary.  Thus quantum analogues follow as a corollary of our classical results.
One application of these results is an alternate proof of the improved strong subadditivity inequality of Li and Winter \cite{LW12}.


\medskip
\noindent
{\bf Adaptive measurements.} The main results in our paper show that certain variants of hypothesis testing are no more difficult than the original problem.  Namely, in the classical case, we can allow an adversary to adaptively change the distribution without decreasing the hypothesis testing exponent, and in the quantum case, we can allow entangled states (under some conditions) while again achieving the same performance.  A natural complementary question is whether hypothesis testing rates can be improved by allowing the distinguisher a broader family of tests.  For example, classically one could consider the problem of distinguishing between two channels (stochastic maps) instead of between two probability distributions, and allowing the distinguisher to adaptively change the inputs to those channels.  In the quantum setting, one might consider the problem of distinguishing $\rho^{\ot n}$ from $\sigma^{\ot n}$ using entangled and/or adaptive measurements.

This sort of adaptivity often does not help.
When distinguishing two classical channels, there is no advantage to using varying inputs in the asymptotic case~\cite{Hayashi-disc}.  On the other hand, in the quantum case, when given $n$ copies of a state, entangled measurements across the $n$ copies  {\em can} improve the hypothesis-testing rate (see \eq{strict-super} and the surrounding discussion).  However, if measurements are forced to be separable across the $n$ copies, then adaptivity is again of no help~\cite{Hayashi-disc} (see also \cite[Section 3.5]{Hayashi-book}).  Thus the results in \cite{Hayashi-disc,Hayashi-book} concern quite a different model (adaptivity of the tester and not of the adversary), and are thus incomparable to ours.
Note that we also consider a different notion of separability, corresponding to cuts of the form $A_1\ldots A_n : B_1\ldots B_n$ instead of $A_1B_1 : A_2B_2 : \cdots : A_nB_n$.





\section{Hypothesis testing against an adaptive adversary}
\label{sec:adaptive}

\subsection{Asymmetric hypothesis testing}
\label{sec:asymh}

Fix two distributions $p$ and $q$ over a finite domain $\Omega$.
Given i.i.d. samples $X_1, X_2, \ldots, X_n$ from a distribution
$r \in \{p,q\}$, the goal is to design a test
which distinguishes the two possibilities based on the sample.
The classical {\em Chernoff-Stein Lemma} characterizes
the optimal exponential rate of error decay achievable in the one-sided error setting.

Consider any acceptance region $A_n \subseteq \Omega^n$ and the corresponding
error probabilities $\alpha_n = p^n({A_n^c})$ and $\beta_n =
q^n(A_n)$, where we use $S^c$ to denote the complement of a set $S$, and $p^n,q^n$ denote $n$ i.i.d. copies of $p,q$ respectively.
Then for $0 < \eps < 1$, define
$$
\beta_n^{\eps} \defeq \min_{\stackrel{A_n \subseteq \Omega^n}{\alpha_n \leq \eps}} \beta_n\,,
$$
and denote the optimal error exponent
$$
\hyp^\eps(p,q) \defeq 
\lim_{n\to\infty} \frac{-\log \beta_n^{\eps}}{n} \,.
$$

The following well-known lemma characterizes $\hyp^\eps$ in terms of the relative entropy
(see, e.g., Theorem 11.8.3 of \cite{CT06}).

\begin{lem}[Chernoff-Stein Lemma]
\label{lem:stein}
Consider any two distributions $p$ and $q$ over a finite domain $\Omega$.
Then $\hyp^\eps(p,q) = \kl{p}{q}$ for any $\eps\in (0,1)$.
\end{lem}

Here, $\kl{p}{q}$ is the {\em relative entropy,} given by
$$
\kl{p}{q} \defeq \sum_{x \in \Omega} p(x) \log \frac{p(x)}{q(x)}\,,
$$
and we take $\kl{p}{q} \defeq \infty$ when there is an $x \in \Omega$ such that $p(x) \neq 0$ but $q(x)=0$.

\medskip
\noindent
{\bf The adaptive setting.}  Suppose now that $P, Q \subseteq \mathbb R^{\Omega}$ are closed, convex sets of probability distributions.
An {\em adaptive $P$-strategy $\hat p$} is a collection of functions $\{\hat p_k : \Omega^{k-1} \to P : k=1,2,\ldots\}$.
Let $\mathcal A(P)$ denote the set of all adaptive $P$-strategies.
For $x \in \Omega^n$, we denote
$$
\hat p(x) \defeq \prod_{k=1}^n \hat p_k(x_1, \ldots, x_{k-1})(x_k)\,.
$$

As before, let $A_n \subseteq \Omega^n$ be an acceptance region, but now we define
$$
\alpha_n \defeq \sup_{\hat p \in \mathcal A(P)} \hat p(A_n^c)\,,
$$
and
$$
\beta_n^{\eps} \defeq \min_{\stackrel{A_n \subseteq \Omega^n}{\alpha_n \leq \eps}} \sup_{\hat q \in \mathcal A(Q)} \hat q(A_n)\,.
$$

For $\eps \in (0,1)$, we denote the {\em adversarial one-sided error exponent} by
$$
\adv^\eps(P,Q) \defeq 
 \lim_{n \to \infty} \frac{- \log \beta_n^{\eps}}{n}\,.
$$
Observe that for single distributions $p,q \in \mathbb R^{\Omega}$, we have
$\adv^\eps(\{p\},\{q\})=\hyp^\eps(p,q)$.

\begin{thm}[Adversarial Chernoff-Stein]\label{thm:adaptive}
Let $\Omega$ be a finite domain.
For any closed, convex sets of probability distributions $P,Q \subseteq \mathbb
R^{\Omega}$ and for any $\eps \in (0,1)$,
we have
\be
\adv^\eps(P,Q) = \min_{p \in P, q \in Q} \kl{p}{q}. \label{eq:adv-rate}
\ee
\end{thm}

Thus in the asymptotic regime, adversarial adaptive hypothesis testing is no harder than the
i.i.d.~setting.  Indeed, when the distributions in $P$ have full support,
the hypothesis test used is a simple
Neyman-Pearson test for $p,q$ minimizing the RHS of \eq{adv-rate}.
This result was previously known in the non-adaptive case, where
it is sometimes referred to as
{\em composite hypothesis testing}~\cite{LevitanM02}.

\begin{proof}[Proof of Theorem~\ref{thm:adaptive}]
Let $p^* \in P$ and $q^* \in Q$ be minimizers of $\kl{p}{q}$ as $p$ and $q$ vary over
$P$ and $Q$, respectively.  Since $P$ and $Q$ are compact and $\kl{p}{q}$ is lower
semi-continuous, such $p^*, q^*$ exist.

By considering non-adaptive strategies that simply play $p^*$ and $q^*$ in each coordinate, one sees that
\ba
\adv^{\eps}(P,Q) & \leq \adv^{\eps}(\{p^*\},\{q^*\}) \twocol{\non \\ & }
= \hyp^{\eps}(p^*, q^*) = \kl{p^*}{q^*},
\label{eq:cs-easy}\ea
where the last
equality is \lemref{stein}.  Thus we need only prove that
\begin{equation}\label{eq:cs-geq}
   \adv^{\eps}(P,Q) \geq \kl{p^*}{q^*}.
\end{equation}
Note that if $\kl{p^*}{q^*}=0$, then \eqref{eq:cs-geq} holds vacuously, and thus we may assume that $P,Q$ are disjoint.

We will establish that \eqref{eq:cs-geq} holds
under the assumption
\begin{equation}\label{eq:support}
   \supp(p) = \supp(q) = \Omega \quad \forall p \in P, q \in Q.
\end{equation}

For any distribution $p$ over $\Omega$, write $p_{\theta} \defeq (1-\theta)p + \theta \frac{\1_{\Omega}}{|\Omega|}$, and
denote $P_{\theta} \defeq (1-\theta) P + \theta \frac{\1_{\Omega}}{|\Omega|}$ and
$Q_{\theta} \defeq (1-\theta) Q + \theta \frac{\1_{\Omega}}{|\Omega|}$.
Since $P$ and $Q$ are disjoint compact convex sets, $P_{\theta}$ and $Q_{\theta}$ are disjoint compact convex sets
for $\theta > 0$ sufficiently small.
At the end of the argument, we will prove the following lemma.

\begin{lemma}\label{lem:cs-eta}
   For every pair of compact convex sets $P,Q \subseteq \mathbb{R}^{\Omega}$ 
   and $\varepsilon \in (0,1)$, it holds that
   \ba
     \adv^{\varepsilon}(P,Q)
&     \geq \limsup_{\theta \to 0} \adv^{\varepsilon}(P_{\theta},Q_{\theta})
\twocol{\nn \\ & }
     = \limsup_{\theta \to 0} \min_{p_\theta\in P_\theta, q_\theta\in Q_\theta}
     D(p_\theta \| q_\theta)
     .
   \label{eq:rate-continuity}\ea
\end{lemma}

To extend this result to general compact convex $P$ and $Q$ without the need for $\theta>0$, we will need to better understand the last term in \eqref{eq:rate-continuity}.
Let $(p_\theta^*,q_\theta^*)\in P_\theta\times Q_\theta$ be a pair achieving the minimum of $\kl{p_\theta}{q_\theta}$.  In fact either these minima are unique or $P\cap Q$ is nonempty thanks to the strict joint convexity of relative entropy (e.g. Theorem 7 of \cite{Ruskai02} ), but we will not use this fact.
%
Since probability distributions over $\Omega$ are a compact set, there exists a positive decreasing sequence $\theta_1 > \theta_2 > \cdots$ such that $\lim_{n\to\infty} (p_{\theta_n}^*,q_{\theta_n}^*)$ exists. Call this limit $(p_0,q_0)$ and observe that $(p_0,q_0)\in P\times Q$.  The lower semi-continuity of  $(p,q) \mapsto \kl{p}{q}$ implies $\liminf_{n\to\infty}\kl{p_{\theta_n}}{q_{\theta_n}}\geq D(p_0\|q_0)$.  Then
\ba\limsup_{\theta \to 0} &\min_{p_\theta\in P_\theta, q_\theta\in Q_\theta}
D(p_\theta \| q_\theta)
\twocol{\nn \\ & }\geq
\liminf_{n\to\infty} \kl{p_{\theta_n}^*}{q_{\theta_n}^*}
\twocol{\nn \\ & }\geq D(p_0\|q_0)
\geq \kl{p^*}{q^*}.\ea
Combined with \eqref{eq:rate-continuity} this establishes
\be     \adv^{\varepsilon}(P,Q)
 \geq \kl{p^*}{q^*}.\ee
 for general compact convex $P$ and $Q$.
 (We believe that $p_\theta^*,q_\theta^*$ are differentiable functions of $\theta$ for $\theta\in (0,1)$ which would permit the above argument to be more direct, but we do not know an obvious proof of this claim.)
\medskip

So let us now assume \eqref{eq:support}.
For $n \in \mathbb N$ and $\delta > 0$, define an acceptance region
\begin{multline*}
A_{n,\delta} = \left\{ x \in \Omega^n :
\twocol{\right. \\ \left. }
  \log \frac{p^*(x_1) p^*(x_2) \cdots p^*(x_n)}{q^*(x_1) q^*(x_2) \cdots q^*(x_n)} \geq n(\kl{p^*}{q^*} - \delta)\right\}\,.
\end{multline*}
Our first goal is to argue that for every $\delta > 0$, we have
\begin{equation}\label{eq:step1}
   \lim_{n \to \infty}\, \inf_{\hat{p} \in \cA(P)} \hat p(A_{n,\delta}) = 1.
\end{equation}
We will then show that for any adaptive $Q$-strategy $\hat q$, we have
\begin{equation}\label{eq:step2}
\hat q(A_{n,\delta}) \leq e^{-n(\kl{p^*}{q^*} - \delta)}\,.
\end{equation}
Once these are proved, letting $\delta \to 0$ yields the desired claim.

Toward proving \eqref{eq:step1}, observe that, for every $\delta > 0$, $\lim_{n \to \infty} (p^*)^n(A_{n,\delta}) = 1$ by the law of large numbers.
The following lemma will allow us to show that the same is true for $\hat p \in \cA(P)$.

\begin{lemma}\label{lem:pythag}
   If \eqref{eq:support} holds, then for any $p \in P$,
$$
\sum_{x \in \Omega} p(x) \log \frac{p^*(x)}{q^*(x)} \geq \sum_{x \in \Omega} p^*(x) \log \frac{p^*(x)}{q^*(x)}\,.
$$
\end{lemma}

\begin{proof}
By Theorem 11.6.1 in \cite{CT06}, we have
$$
\kl{p}{q^*} \geq \kl{p}{p^*} + \kl{p^*}{q^*}\,.
$$
Observing that $\kl{p}{q^*} - \kl{p}{p^*} = \sum_{x \in \Omega} p(x) \log \frac{p^*(x)}{q^*(x)}$,
we see that this is precisely the desired inequality.
\end{proof}

Now, for $x \in \Omega$, define $L(x) \seteq \log \frac{p^*(x)}{q^*(x)}$.
Note that \eqref{eq:support} implies
%
\begin{equation}\label{eq:diffs}
   m = m(p^*,q^*) \seteq \max \left\{ |L(x)| : x \in \Omega \right\} < \infty.
\end{equation}
Moreover, \lemref{pythag} yields
\begin{equation}\label{eq:mg}
   \E_p [L(x)] \geq \E_{p^*}[L(x)] = \kl{p^*}{q^*}, \qquad \forall p \in P.
\end{equation}

Let $\hat{p} \in \cA(P)$ denote some adaptive $P$-strategy.
Consider a sequence of random variables $\{X_k\}$ distributed
according to $\hat{p}$ (i.e., $X_{k}$ is sampled according
to the measure $\hat{p}_k(X_1, X_2, \ldots, X_{k-1}) \in P$),
and the corresponding martingale difference sequence
$$D_k \defeq L(X_k) - \E [L(X_k) \mid X_1, \ldots, X_{k-1}]\,.$$
(Recall that the defining property of a martingale difference sequence is that $\E[|D_k|]$
is finite and $\E[D_k
\mid X_1,\ldots,X_{k-1}]=0$ for any $X_1,\ldots,X_{k-1}$.)
Since the differences are uniformly bounded (cf. \eqref{eq:diffs}),
orthogonality of martingale difference sequences yields
\[
   \E\left(\sum_{k=1}^n D_k\right)^2 = \sum_{k=1}^n \E[D_k^2] \leq 4 m^2 n\,.
\]

Chebyshev's inequality then implies that for any $\delta > 0$,
\begin{equation}\label{eq:cheby}
   \pr\left(\sum_{k=1}^n D_k \geq -\delta n\right) \geq 1 - \frac{4 m^2}{\delta^2} \frac{1}{n}\,.
\end{equation}
On the other hand, \eqref{eq:mg} implies that for each $k$, one has $\E [L(X_k) \mid X_1, \ldots, X_{k-1}] \geq \kl{p^*}{q^*}$.
Combining this with \eqref{eq:cheby} yields
\ba
\hat{p}(A_{n,\delta}) &
= \pr\left(\sum_{k=1}^n L(X_k) \geq n(\kl{p^*}{q^*}-\delta)\right)
\twocol{\nn \\ &  } \geq
\pr\left(\sum_{k=1}^n D_k \geq -\delta n\right) \geq 1 - 
\frac{4 m^2/\delta^2}{n}\,.
\label{eq:cheby2}\ea
Noting that the latter expression goes to $1$ as $n\to \infty$ (uniformly in $\hat{p}$) 
confirms \eqref{eq:step1}.
We now turn to verifying \eqref{eq:step2}.

\begin{lem}\label{lem:tech}
For any $q \in Q$, we have
$$
\sum_{x \in \Omega} q(x) \frac{p^*(x)}{q^*(x)} \leq 1\,.
$$
\end{lem}

\begin{proof}
For $\lambda \in [0,1]$, write
$q_{\lambda} = \lambda q + (1-\lambda) q^*$.  Since $q^*$ is the minimizer
of $\kl{p^*}{q}$ for $q$ in the convex set $Q$, we know that the derivative of $\kl{p^*}{q_{\lambda}}$
at $\lambda=0$ is non-negative.

Calculate
\bas
\frac{d}{d\lambda} &\kl{p^*}{q_{\lambda}} =
\sum_{x \in \Omega}p^*(x)
\frac{d}{d\lambda} \log \frac{p^*(x)}{q_{\lambda}(x)}
\\
&= - \sum_{x \in \Omega} p^*(x) \frac{d}{d\lambda} \log \left(\frac{\lambda q(x) + (1-\lambda) q^*(x)}{p^*(x)}\right)
\\
&= - \sum_{x \in \Omega} p^*(x) \frac{q(x)-q^*(x)}{\lambda q(x) + (1-\lambda) q^*(x)}\,.
\eas
Using the fact that the derivative is non-negative at $\lambda=0$ yields
$$
\sum_{x \in \Omega} \frac{p^*(x) q^*(x)}{q^*(x)} \geq \sum_{x\in \Omega} \frac{p^*(x) q(x)}{q^*(x)}\,,
$$
but the left-hand side is equal to 1, yielding the desired result.
\end{proof}

With the preceding lemma in hand, we finish the proof of \eqref{eq:step2}.
Fix some adaptive $Q$-strategy $\hat q$.
By Markov's inequality,
\begin{equation}\label{eq:markov}
\hat q(A_{n,\delta}) \leq e^{-n (\kl{p^*}{q^*} - \delta)}\, \E_{\hat q} \left[\frac{p^*(x_1) \cdots p^*(x_n)}{q^*(x_1) \cdots q^*(x_n)}\right].
\end{equation}
We now use the fact that,
by \lemref{tech}, the sequence of likelihood ratios $\prod_{i=1}^n
\frac{p^*(x_i)}{q^*(x_i)}$ is a supermartingale with respect to $\hat q$.  (Recall that a
sequence $X_1,X_2,\ldots$ is a supermartingale if $\E[X_n|X_1,\ldots,X_{n-1}] \leq
X_{n-1}$ for all choices of $n$ and $X_1,\ldots,X_{n-1}$.)
In particular,
\ba
 \E_{\hat q} & \left[\frac{p^*(x_1) \cdots p^*(x_n)}{q^*(x_1) \cdots q^*(x_n)}\right]
\twocol{\nn \\ & } =
 \E_{\hat q} \left[\frac{p^*(x_1) \cdots p^*(x_{n-1})}{q^*(x_1) \cdots q^*(x_{n-1})} \E_{\hat q_n(x_1, x_2, \ldots, x_{n-1})} \frac{p^*(x)}{q^*(x)} \right] \nonumber\\
 &\leq 
 \E_{\hat q} \left[\frac{p^*(x_1) \cdots p^*(x_{n-1})}{q^*(x_1) \cdots q^*(x_{n-1})} \right] \nonumber\\
 &\leq  \cdots \nonumber \\
 &\leq  1, \label{eq:cs-sm}
\ea
where in the second line we have applied \lemref{tech} to the distribution $\hat q_n(x_1, x_2, \ldots, x_{n-1}) \in Q$,
and then we have continued by induction.  Combining this with \eqref{eq:markov} completes our verification of \eqref{eq:step2}
and hence our proof of the theorem.

The proof of the theorem then follows from Lemma~\ref{lem:cs-eta}.
\end{proof}
\begin{proof}[Proof of Lemma~\ref{lem:cs-eta}]
  Let $(p^*_{\theta},q^*_{\theta}) \in P_{\theta} \times Q_{\theta}$ be a pair minimizing
  $\kl{p_{\theta}}{q_{\theta}}$ over $(p_{\theta},q_{\theta}) \in P_{\theta} \times Q_{\theta}$.
  Note that $\kl{p_{\theta}^*}{q_{\theta}^*} < \infty$ for $\theta > 0$.

  For any $\delta > 0$,
  define the acceptance region
\begin{multline*}
    A^{\theta}_{n,\delta} \seteq 
    \left\{ x \in \Omega^n :
      \twocol{\right. \\ \left.}
      \log \frac{p_{\theta}^*(x_1) p_{\theta}^*(x_2) \cdots p_{\theta}^*(x_n)}{q_{\theta}^*(x_1) q_{\theta}^*(x_2) \cdots q_{\theta}^*(x_n)} \geq n\left(\kl{p_{\theta}^*}{q_{\theta}^*} - \delta\right)\right\}\,.
\end{multline*}

Let $\hat{p}$ and $\hat{q}$ denote an adaptive $P$-strategy and $Q$-strategy, respectively, and
define adaptive $P_{\theta}$ and $Q_{\theta}$-strategies by
\begin{align*}
   (\hat{p}_{\theta})_k(x_1, \ldots, x_{k-1}) &= (\hat{p}_k(x_1,\ldots, x_{k-1}))_{\theta} \\
   (\hat{q}_{\theta})_k(x_1, \ldots, x_{k-1}) &= (\hat{q}_k(x_1,\ldots, x_{k-1}))_{\theta}.
\end{align*}
In other words, $x_k$ is obtained by sampling from $\hat{p}_k(x_1,\ldots, x_{k-1})$ with probability $1-\theta$ and from the uniform distribution with probability $\theta$.  This implies that
\be \frac 12 \|\hat{p} - \hat{p}_{\theta}\|_1\leq n\theta
\label{eq:p-hat-theta-dist}.\ee

   Now we use the super martingale property \eqref{eq:cs-sm} to obtain
   \begin{align}
      1 & \geq \E_{\hat{q}_{\theta}} \left[\prod_{i=1}^n \frac{p_{\theta}^*(x_i)}{q_{\theta}^*(x_i)}\right] \\
      &
      =\E_{\hat{q}_{\theta}} \left[\E_{(\hat{q}_{\theta})_n(x_1,\ldots,x_{n-1})} \left[\frac{p_{\theta}^*(x_n)}{q_{\theta}^*(x_n)}\right] \prod_{i=1}^{n-1} \frac{p_{\theta}^*(x_i)}{q_{\theta}^*(x_i)}\right] \nonumber \\
      &\geq
      (1-\theta) \E_{\hat{q}_{\theta}} \left[\E_{\hat{q}_n(x_1,\ldots,x_{n-1})} \left[\frac{p_{\theta}^*(x_n)}{q_{\theta}^*(x_n)}\right] \prod_{i=1}^{n-1} \frac{p_{\theta}^*(x_i)}{q_{\theta}^*(x_i)}\right]  \nonumber\\
      &\geq \cdots \nonumber\\
      &\geq (1-\theta)^n \E_{\hat{q}} \left[\prod_{i=1}^n \frac{p_{\theta}^*(x_i)}{q_{\theta}^*(x_i)}\right]  \nonumber \\
      &\geq (1-\theta)^n e^{n (\kl{p_{\theta}^*}{q_{\theta}^*} - \delta)} \hat{q}(A_{n,\delta}^{\theta}). \label{eq:cs-sm2}
   \end{align}
   Next, we use first \eqref{eq:p-hat-theta-dist}, then
 $m(\cdot,\cdot)$ from \eqref{eq:diffs} and finally  the bound
$m(p_{\theta}^*,q_{\theta}^*)\leq \log( |\Omega|/\theta)$ to establish
   \begin{align}
     \hat{p}(A_{n,\delta}^{\theta})
     &\geq \hat{p}_{\theta}(A_{n,\delta}^{\theta}) - \theta n
     \\ &\geq 1 - \frac{4m(p^*_{\theta},q^*_{\theta})^2}{\delta^2}\frac{1}{n} - \theta n,
     \\ & \geq 1 - \frac{4\log^2( |\Omega|/\theta)}{\delta^2}\frac{1}{n} - \theta n,
   \end{align}

This bound approaches 1 as long as $\theta$ decreases at an appropriate rate with $n$, say by
taking $\theta=1/n^2$. As a result,
   \[
      \lim_{n \to \infty} \inf_{\hat{p} \in \cA(P)} \hat{p}\left(A_{n,\delta}^{1/n^2}\right) = 1.
   \]
   Combining this with \eqref{eq:cs-sm2} shows that along the sequence of acceptance regions
   $\{A_{n,\delta}^{1/n^2}\}$, we have
   \begin{align*}
      \adv^{\eps}\twocol{&} (P,Q) \twocol{\nn \\ \! \!} &\geq \lim_{n \to \infty} \left(\kl{p_{1/n^2}^*}{q_{1/n^2}^*} - \delta + \log(1-n^{-2})\right) \\
                       &\geq \lim_{n \to \infty} \left(\adv^{\eps}(P_{1/n^2},Q_{1/n^2}) - \delta\right),
   \end{align*}
   where the second inequality follows from \eqref{eq:cs-easy}.
   Now taking $\delta \to 0$ completes the proof.
\end{proof}

\subsection{Chernoff information and symmetric hypothesis testing}

Suppose again that we have two distributions $p$ and $q$ over a finite domain $\Omega$.
We also have $n$ i.i.d. samples $X_1, X_2, \ldots, X_n$ from a distribution $r \in \{p,q\}$,
and a Bayesian hypothesis:  The samples come from $p$ with probability $\pi_p$
and from $q$ with probability $\pi_q$.  Consider a test $T_n \subseteq \Omega^n$.
If $(X_1, X_2, \ldots, X_n) \in T_n$, we declare that the sample came from $p$.

Our goal is to minimize the expected error
$$
\delta_n(T_n) \defeq \pi_p\, p^n(T_n^c) + \pi_q\, q^n(T_n)\,.
$$
In this case, the best achievable error exponent is
$$
\gamma(p,q) \defeq \lim_{n \to \infty} - \frac{1}{n} \min_{T_n \subseteq \Omega^n} \log \delta_n(T_n)\,.
$$
Observe that the constants $\pi_p$ and $\pi_q$ do not affect $\gamma(p,q)$.

For $\lambda \in (0,1)$, let us define
$$\Gamma^{\lambda}(p,q) \defeq -\log \sum_{x \in \Omega, p(x)q(x)>0} p(x)^{\lambda} q(x)^{1-\lambda}\,,$$
and
\begin{equation}\label{eq:gamma-star}
   \Gamma^*(p,q) \defeq \sup_{\lambda \in (0,1)} \Gamma^{\lambda}(p,q)\,.
\end{equation}
We have the following characterization due to Chernoff (see, e.g.,
Theorem 11.9.1 of \cite{CT06}).

\begin{thm}\label{thm:precher}
For any distributions $p$ and $q$ on $\Omega$,
one has
\[
   \gamma(p,q) = \Gamma^*(p,q)\,.
\]
Moreover, if $\supp(p)=\supp(q)$\footnote{The statement of Theorem 11.9.1 in \cite{CT06} does not include
the condition that $\supp(p)=\supp(q)$, but as was pointed out to us by an anonymous referee,
there are examples where the theorem is false without this assumption.}, then one has
$$
   \gamma(p,q) = \Gamma^*(p,q) = \kl{r}{p} = \kl{r}{q}\,,
$$
where $r$ is the distribution given by
$$
r(x) \defeq \frac{p(x)^{\lambda(p,q)} q(x)^{1-\lambda(p,q)}}{\sum_{y \in \Omega} p(y)^{\lambda(p,q)} q(y)^{1-\lambda(p,q)}}\,,
$$
and $\lambda(p,q)$ is the unique value of $\lambda \in (0,1)$ achieving the supremum in \eqref{eq:gamma-star}.
\end{thm}

We will prove a corresponding theorem in the adaptive setting.  To this end consider again two closed, convex sets of
distributions $P, Q \subseteq \mathbb R^{\Omega}$.  Define the {\em adversarial two-sided error exponent}
\bas
\advg\twocol{&}(P,Q) \defeq \twocol{\\ & } \lim_{n \to \infty} - \frac{1}{n} \min_{T_n \subseteq \Omega^n} \max_{\hat p,\hat q} \log \left(\hat p(T_n^c) + \hat q(T_n)\right)\,
\eas
where the maximum is over all adaptive $P$-strategies $\hat p$ and adaptive $Q$-strategies $\hat q$.

\def\ThmAdaptiveChernoff{
For any finite domain $\Omega$ and closed, convex sets of distributions $P,Q \subseteq \mathbb R^{\Omega}$,
we have
\begin{equation}\label{eq:adaptive-chernoff}
\advg(P,Q) = \min_{p \in P, q \in Q} \Gamma^*(p,q)\,.
\end{equation}}
\begin{thm}[Adversarial Chernoff's Theorem]
\label{thm:chernoff}
\ThmAdaptiveChernoff
\end{thm}

\begin{proof}
We may assume that $P$ and $Q$ are compact; the general case can be reduced to this one
by considering exhaustions of $P$ and $Q$ by compact convex sets.
Assume $P$ and $Q$ are disjoint, since otherwise $\advg(P,Q)=\min_{p\in P,q\in Q} \Gamma^*(p,q)=0$.
Let $p^* \in P, q^* \in Q$ be some pair that minimizes $\Gamma^*(p,q)$ over $p \in P, q\in Q$.
First,
we have $$\advg(P,Q) \leq \advg(\{p^*\},\{q^*\}) = \gamma(p^*,q^*) = \Gamma^*(p^*,q^*)\,,$$
where the latter equality is given by Theorem \ref{thm:precher}.  Thus we
are left to prove $\advg(P,Q) \geq \Gamma^*(p^*,q^*)$.

\medskip

Let us first 
assume that $\supp(p)=\supp(q)=\Omega$ for all $p \in P$ and $q \in Q$.
After the argument, we will reduce the general case to this one.
Consider $p \neq q$.
Define $F_{p,q} : [0,1] \to \mathbb{R}$ by
\[
   F_{p,q}(\lambda) \seteq \sum_{x\in \Omega} p(x)^{\lambda} q(x)^{1-\lambda}\,,
\]
and calculate
\begin{align*}
   F'_{p,q}(\lambda) &= \sum_{x\in\Omega} p(x)^{\lambda} q(x)^{1-\lambda} \log \frac{p(x)}{q(x)}\,, \\
   F''_{p,q}(\lambda) &= \sum_{x \in \Omega} p(x)^{\lambda} q(x)^{1-\lambda} \left(\log \frac{p(x)}{q(x)}\right)^2.
\end{align*}
Since $p \neq q$,
$F''_{p,q}(\lambda) > 0$ for all $\lambda \in (0,1)$.
Since additionally $\supp(p)=\supp(q)$,
\begin{align*}
   F_{p,q}'(0) &= \kl{p}{q} > 0 \\
   F_{p,q}'(1) &= -\kl{q}{p} < 0\,.
\end{align*}
We conclude that $F_{p,q}(\lambda)$ is minimized at a unique value $\lambda \in (0,1)$.
Denote this value by $\lambda(p,q)$ and observe that $\Gamma^*(p,q) = \Gamma^{\lambda(p,q)}(p,q)$.
Let $\lambda^* \defeq \lambda(p^*,q^*)$.

Define now
$$
T_n \defeq \left\{ x \in \Omega^n : \prod_{i=1}^n p^*(x_i) \geq \prod_{i=1}^n q^*(x_i) \right\}\,.
$$
Fix also an adaptive $P$-strategy $\hat p$ and an adaptive $Q$-strategy $\hat q$.  We will show that
\begin{equation}\label{eq:twosided}
   \Gamma^*(p^*,q^*) \leq \lim_{n \to \infty} \frac{-\log(\hat p(T_n^c) + \hat q(T_n))}{n}\,.
\end{equation}

We will need to employ the following easy variant of the ``envelope theorem.''

\begin{lemma}\label{lem:envelope}
Consider a differentiable function $f : [0,1]^2 \to \mathbb R$.  Define $V(t) = \inf_{\lambda \in [0,1]} f(\lambda,t)$ and
suppose that for every $t \in [0,1]$, there is a unique $\lambda^*(t) \in (0,1)$ such that $V(t) = f(\lambda^*(t), t)$.
If $\lambda^*$ is differentiable at $t \in [0,1]$, then $V'(t) = f_2(\lambda^*(t), t)$ where $f_2$ is the partial derivative
of $f$ with respect to its second argument.
\end{lemma}

\begin{proof}
Let $f_1$ denote the partial derivative of $f$ with respect to its first argument.
Writing $V(t) = f(\lambda^*(t),t)$ and applying the chain rule yields
$$V'(t) =f_2(\lambda^*(t), t) + f_1(\lambda^*(t),t) \frac{d}{dt} \lambda^*(t).$$
The second term is zero because $f_1(\lambda^*(t),t)=0$ by optimality of $\lambda^*(t)$.
\end{proof}

\begin{rem}\label{rem:envelope}
Observe that if $f(\lambda,t)$ has $\frac{\partial^2}{\partial \lambda^2} f(\lambda,t) > 0$ for some $t \in [0,1]$,
then $\lambda^*(t)$ is the unique solution of $\frac{\partial}{\partial \lambda} f(\lambda,t) = 0$
and is differentiable by the implicit function theorem.  Note that the assumptions of \lemref{envelope}
can be relaxed considerably; see, e.g., \cite[Ch. 3]{milgrom2004putting}.
\end{rem}

This allows us to prove the following.

\begin{lem}\label{lem:chopt}
For any distribution $q \in Q$, one has
$$
\sum_{x \in \Omega} q(x) \frac{p^*(x)^{\lambda^*}}{q^*(x)^{\lambda^*}} \leq \sum_{x \in \Omega} q^*(x) \frac{p^*(x)^{\lambda^*}}{q^*(x)^{\lambda^*}}\,.
$$
\end{lem}

\begin{proof}
For $t \in [0,1]$, define a distribution $q_t \defeq t q + (1-t)q^* \in Q$.
Moreover, define a function $f : [0,1]^2 \to \mathbb R$ by $$f(\lambda,t) \defeq F_{p^*,q_t}(\lambda)\,.$$
As we have already observed, for every fixed value of $t \in [0,1]$, it holds that
$\lambda(p^*,q_t) \in (0,1)$ is the unique minimizer of $f(\lambda,t)$.

Let $f_2$ be the partial derivative of $f$ in its second argument; then one computes:
$$
f_2(\lambda,t) = \sum_{x \in \Omega} (q(x)-q^*(x)) (1-\lambda) q_t(x)^{-\lambda} p^*(x)^{\lambda}\,.
$$
If we let $V(t) = \min_{\lambda \in (0,1)} f(\lambda,t)$, then optimality of $q^*$ implies $V'(0) \leq 0$.
But now \lemref{envelope} (in conjunction with Remark \ref{rem:envelope}) yields
\begin{eqnarray*}
0 &\geq& V'(0) = f_2(\lambda^*, 0) \\
&=& \sum_{x \in \Omega} (q(x)-q^*(x)) (1-\lambda^*) q^*(x)^{-\lambda^*} p^*(x)^{\lambda^*}\,.
\end{eqnarray*}
Rearranging yields the desired claim.
\end{proof}

The preceding lemma shows that the sequence $\prod_{i=1}^n \frac{p^*(x_i)^{\lambda^*}}{q^*(x_i)^{\lambda^*}}$ is a supermartingale with
respect to $\hat q$.
Thus we can write
\ba
\E_{\hat q}\twocol{&} \left[\prod_{i=1}^n
  \frac{p^*(x_i)^{\lambda^*}}{q^*(x_i)^{\lambda^*}}\right]
\twocol{\nn \\} &= \nonumber
\E_{\hat q}\left[\prod_{i=1}^{n-1} \frac{p^*(x_i)^{\lambda^*}}{q^*(x_i)^{\lambda^*}} \E_{\hat q_n(x_1, \ldots, x_{n-1})} \frac{p^*(x_n)^{\lambda^*}}{q^*(x_n)^{\lambda^*}}\right] \nonumber \\
&\leq e^{-\Gamma^*(p^*,q^*)} \E_{\hat q}\left[\prod_{i=1}^{n-1} \frac{p^*(x_i)^{\lambda^*}}{q^*(x_i)^{\lambda^*}}\right] \nonumber
\\
&\leq \cdots \nonumber \\
&\leq e^{-n\Gamma^*(p^*,q^*)}\,, \label{eq:sm}
\ea
where in the second line we have used \lemref{chopt} along with
the fact that $q=\hat q_n(x_1, \ldots, x_{n-1}) \in Q$,
and then we have continued by induction.

By Markov's inequality, this implies $\hat q(T_n) \leq e^{- n  \Gamma^*(p^*,q^*)}$.
By the symmetry of the preceding argument with respect to $P$ and $Q$, the same bound of $\hat p(T_n^c) \leq e^{- n \Gamma^*(p^*,q^*)}$ holds for $\hat p$.
Combining these yields $\advg(P,Q) \geq \Gamma^*(p^*,q^*)$, completing the proof.

\medskip
\noindent
{\bf General $P$ and $Q$.}
Let us recall from Section~\ref{sec:asymh} the notation $p_{\theta}$ for $p \in \Omega$, and the sets
$P_{\theta}$ and $Q_{\theta}$.

\begin{lemma}\label{lem:j1}
      For any $\theta > 0$ sufficiently small, it holds that
      \[
         \advg(P,Q) \geq \advg(P_{\theta},Q_{\theta}) - \theta\,.
      \]
   \end{lemma}

   \begin{lemma}\label{lem:j2}
   It holds that
   \[
      \liminf_{\theta \to 0} \min_{(p,q) \in P \times Q} \Gamma^*(p_{\theta},q_{\theta}) \geq \min_{(p,q) \in P \times Q} \Gamma^*(p,q)\,.
   \]
\end{lemma}

Let us first use them to complete the proof of our desired result for general $P$ and $Q$
using the result for $P_{\theta}$ and $Q_{\theta}$.  Employ Lemma \ref{lem:j1} and then Lemma \ref{lem:j2} to write:
\ba
\advg(P,Q) & \geq \liminf_{\e \to 0} \advg(P_{\theta},Q_{\theta})
\twocol{\nn \\ & }= 
      \liminf_{\theta \to 0} \min_{(p,q) \in P \times Q}
      \Gamma^*(p_{\theta},q_{\theta})
\twocol{\nn \\ & }      \geq \min_{(p,q) \in P \times Q} \Gamma^*(p,q)\,.
\ea
This concludes the proof of \thmref{chernoff}, modulo the proofs of Lemmas \ref{lem:j1}
and \ref{lem:j2}.
\end{proof}

\begin{proof}[Proof of Lemma \ref{lem:j1}]
      Let $(p^*_{\theta},q^*_{\theta}) \in P_{\theta} \times Q_{\theta}$ be a pair minimizing
      $\Gamma^*(p_{\theta},q_{\theta})$ over $(p_{\theta},q_{\theta}) \in P_{\theta} \times Q_{\theta}$,
      and let $\lambda^*$ denote their optimal exponent.
      Define the test
      \begin{multline*} T_n(p_{\theta}^*,q_{\theta}^*) \seteq
        \twocol{\\} \{ x \in \Omega^n : p_{\theta}^*(x_1) \cdots p_{\theta}^*(x_n) \geq q_{\theta}^*(x_1) \cdots q_{\theta}^*(x_n) \}\,.\end{multline*}

      Let $\hat{q}$ denote an adaptive $Q$-strategy.
      We define an adaptive $Q_{\theta}$-strategy
   $\hat{q}_{\theta}$ by
   \[
      (\hat{q}_{\theta})_k(x_1, \ldots, x_{k-1}) = (\hat{q}_k(x_1,\ldots, x_{k-1}))_{\theta}.
   \]
   Then we have:
   \begin{align*}
      \exp & \left(-n \Gamma^*(p_{\theta}^*, q_{\theta}^*)\right)
      \geq \E_{\hat{q}_{\theta}}\left[\prod_{i=1}^n \frac{p_{\theta}^*(x_i)^{\lambda^*}}{q_{\theta}^*(x_i)^{\lambda^*}}\right]
      \twocol{\\ & } =
      \E_{\hat{q}_{\theta}} \left[ \E_{(\hat{q}_{\theta})_n(x_1,\ldots,x_{n-1})}\left[\frac{p_{\theta}^*(x_n)^{\lambda^*}}{q_{\theta}^*(x_n)^{\lambda^*}}\right]
      \prod_{i=1}^{n-1} \frac{p_{\theta}^*(x_i)^{\lambda^*}}{q_{\theta}^*(x_i)^{\lambda^*}}\right],
\end{align*}
where the first inequality uses the supermartingale inequality \eqref{eq:sm}.
Moreover, for every $(x_1,\ldots,x_{n-1}) \in \Omega^{n-1}$,
\begin{align*}
   \E_{(\hat{q}_{\theta})_n(x_1,\ldots,x_{n-1})} & \left[\frac{p_{\theta}^*(x_n)^{\lambda^*}}{q_{\theta}^*(x_n)^{\lambda^*}}\right] 
\twocol{\\}\geq \twocol{&} (1-\theta)\, \E_{\hat{q}_n(x_1,\ldots,x_{n-1})}\left[\frac{p_{\theta}^*(x_n)^{\lambda^*}}{q_{\theta}^*(x_n)^{\lambda^*}}\right],
\end{align*}
thus continuing inductively yields

\begin{align*}
   \exp\twocol{&}\left(-n \Gamma^*(p_{\theta}^*, q_{\theta}^*)\right) \onecol{&}\geq
   \E_{\hat{q}_{\theta}} \left[\prod_{i=1}^n \frac{p_{\theta}^*(x_i)^{\lambda^*}}{q_{\theta}^*(x_i)^{\lambda^*}}\right] \\ 
   &\geq (1-\theta)^n\,
      \E_{\hat{q}} \left[\prod_{i=1}^n \frac{p_{\theta}^*(x_i)^{\lambda^*}}{q_{\theta}^*(x_i)^{\lambda^*}}\right] \\
               &\geq
\onecol{   (1-\theta)^n 
     \hat{q}
               \left(\left\{ (x_1,\ldots,x_n) \in \Omega^n
     : p_{\theta}^*(x_1)\cdots p_{\theta}^*(x_n) \geq q_{\theta}^*(x_1)\cdots q_{\theta}^*(x_n)\right\}\right)  \\ 
   & = } (1-\theta)^{n}\, \hat{q}\left(T_n(p_{\theta}^*, q_{\theta}^*)\right).
\end{align*}
Doing the symmetric analysis with an adaptive $P$-strategy yields
\bas
\Gamma^* & (p_{\theta}^*, q_{\theta}^*) \leq
\twocol{& \\ &  } \theta - \frac{1}{n} \max_{\hat{p},\hat{q}} \left[\log \left(\vphantom{\bigoplus}\frac{\hat{q}(T_n(p_{\theta}^*, q_{\theta}^*))+
   \hat{p}(T_n^c(p_{\theta}^*, q_{\theta}^*))}{2}\right)\vphantom{\bigoplus}\right],
\eas
and then taking the limit as $n \to \infty$ gives
\[\advg(P_{\theta},Q_{\theta}) = \Gamma^*(q_{\theta}^*, p_{\theta}^*)
  \leq \advg(P,Q) + \theta\,.\qedhere
\]
\end{proof}

\begin{proof}[Proof of Lemma \ref{lem:j2}]
   Let $\{ (p^n,q^n) \in P \times Q : n=1,2,\ldots\}$ denote a sequence of distributions,
   and consider a sequence $\{\theta_n\}$ with $\theta_n \to 0$ as $n \to \infty$.
   Since $P \times Q$ is compact, we may pass to a subsequence where $(p^n,q^n)$ converges.
   Let $(\bar{p},\bar{q}) \in P \times Q$ be the limit.
   Note that $(\bar{p},\bar{q})$ is also a limit of the sequence $\{ (p^n_{\theta_n},q^n_{\theta_n}) \}$.

   Observe now that $\Gamma^*(p,q)$ is a supremum of continuous functions, and thus $(p,q) \mapsto \Gamma^*(p,q)$
   is lower semi-continuous.
   This implies that
   \[
      \lim_{n \to \infty} \Gamma^*(p^n_{\theta_n}, q^n_{\theta_n}) \geq \Gamma^*(\bar{p},\bar{q})\,,
   \]
   completing the proof.
\end{proof}

\section{Distinguishing quantum states with restricted measurements}
\label{sec:quantum}

A central problem in quantum information is to distinguish between a pair of quantum
states $\rho$ and $\sigma$.    As usual, there is a tradeoff between
errors of type 1 and 2, i.e., mistaking $\rho$ for $\sigma$ and vice
versa.  The quantum Neyman-Pearson lemma states that the optimal
tradeoff curve between errors of type 1 and 2 is achieved by choosing
$$\mathcal{M} = \{\theta\rho - \sigma \geq 0\},$$
for some $\theta\geq 0$, where $\{X \geq 0\}$ denotes the projector onto the eigenvectors of
$X$ with nonnegative eigenvalue.
The estimation strategy is  then to perform the measurement
$\{\mathcal{M},I-\mathcal{M}\}$ and guess $\rho$ upon obtaining the
outcome corresponding to POVM element $\mathcal{M}$ or $\sigma$
upon obtaining the outcome corresponding to $I-\mathcal{M}$.

{\em Remark on terminology:} We briefly introduce some notation here, and additional
background and definitions for the reader unfamiliar with quantum information theory can
be found in Appendix \ref{sec:background}.  The finite domain $\Omega$ from
\secref{adaptive} is replaced with $V=\bbC^d$ with the standard
Euclidean inner product, and we denote the
set of density operators on $V$ by $\cD(V)$.  Let $\cL(V)$ denote linear operators on $V$
and let $E(V) = \{ \cM \in L(V) : 0 \leq \cM \leq I\}$ be the space of POVM elements.  A {\em
  measurement} $\cM = (\cM_1,\cM_2,\ldots)$ is a collection of POVM elements that sum to $I$,
and $\cM(\rho)=(\tr(\cM_1\rho), \tr \cM_2\rho),\ldots)$ refers to the probability distribution
of measurement outcomes resulting from applying $\cM$ to $\rho$.  For our purposes we will
consider both two-outcome measurements and measurements with finitely
many nonzero POVM elements. 
Call these sets $E_2(V)$ and $E_{\bbN}(V)$
respectively.  For $E_2(V)$, the measurement $\{\cM,I-\cM\}$ is of course determined by
the first POVM element $\cM$ and so where it is not ambiguous we will
use $\cM$ to refer to the measurement.  Further background on quantum
states and measurements can be found in the appendix.

One well-known case of state distinguishability 
is when $\rho$ and $\sigma$ have prior probabilities
$p$ and $1-p$, respectively, and we wish to minimize the total probability of error.  In
this case the optimal measurement $\mathcal{M}$ is given by $\mathcal{M} = \{p \rho - (1-p)\sigma \geq
0\},$ and the probability of error is $\frac{1-\|p\rho -
  (1-p)\sigma\|_1}{2}$, where $\|\cdot\|_1$ denotes the Schatten
1-norm.   (Here $\cM$ corresponds to guessing ``$\rho$'' and $I-\cM$ to guessing ``$\sigma$''.)
 The familiar {\em trace distance}
$\frac{1}{2}\|\rho-\sigma\|_1$ corresponds to the case $p=1/2$.

\bigskip
\noindent
We modify this basic problem of state
distinguishability in three (simultaneous) ways:
\benum
\item We consider only measurements $\cM$ from some restricted
class $M\subseteq E_2(V)$.
\item  We allow $\rho,\sigma$ to be drawn adversarially from some
       sets $R,S $, respectively.  (This means that an adversary chooses
       $\rho,\sigma$ in each round with knowledge of  all
       previous measurement outcomes.)
\item We consider the asymptotic limit in which $M,R,S$ are replaced
  by families $\famM=(M^1,M^2,\ldots), \famR=(R^1,R^2,\ldots), \famS=(S^1,S^2,\ldots)$ with $M^n, R^n, S^n$ describing measurements and states on $V^{\ot
  n}$. Our goal is then, for each $n$, to find a
  measurement $\cM \in M^n$ that will effectively distinguish any
  state $\rho\in R^n$ from any state $\sigma \in S^n$.
\eenum

These changes render the problem a good deal more abstract, and
introduce a large number of new parameters.  Thus, it may be helpful
to keep in mind a prototypical example that was one of the motivations
for this work.  For some fixed bipartite state $\rho$ over $A \ot B$,
let $R^n$ be the singleton set $\{\rho^{\ot n}\}$, and let $S^n \defeq
\Sep(A^{\ot n} : B^{\ot n})$.  This corresponds to studying the
asymptotic distinguishability of many copies of $\rho$ from a
separable state on the same number of systems.
For this special case, we introduce the notation 
\bas \famrho & := (\{\rho\}, \{\rho^{\ot 2}\}, \ldots) \\
\famSep(A:B) &  := (\Sep(A:B), \Sep(A^{\ot 2}:B^{\ot 2}), \ldots).
\eas
   Where the context is
understood, we will often omit the reference to $A,B$ and simply write
$\Sep$ or $\famSep$.
Finally, we will consider a restricted class of measurements
$\famM$, such as the class of 1-LOCC measurements (as discussed in
\cite{Piani09, BCY11, LW12, BH-local}).


\subsection{Background on restricted quantum measurements}
\label{sec:q-background}

We begin by introducing notation, describing
known results on restricted-measurement distinguishability, and
presenting a few small new results to help clean up the landscape.
In \secref{q-stein}, we describe our restricted-measurement version of
the quantum Stein's Lemma, and in \secref{q-SSA} we give an application to
quantum conditional mutual information.

\subsubsection{Quantum Stein's Lemma}
\label{sec:q-stein-intro}

If $\rho,\sigma$ are density matrices on a space $V$, then the {\em relative
entropy of $\rho$ with respect to $\sigma$} is
\be \kl{\rho}{\sigma} \defeq
\tr \left( \rho(\log \rho - \log\sigma) \right) .\ee
If $\ker(\sigma) \nsubseteq \ker(\rho)$, we take $\kl{\rho}{\sigma} \defeq \infty$.

 Following the classical case, we
define an {\em acceptance operator} $\cM^n \in E(V^{\ot n})$ (analogous to the
acceptance region $T_n$),
with
corresponding error probabilities $\alpha_n =\tr \left( (I-\cM^n) \rho^{\ot n} \right)$
and $\beta_n \defeq \tr \left( \cM^n \sigma^{\ot n} \right)$. Again we can define
$\beta_n^\eps \defeq \min\{\beta_n : \alpha_n \leq \eps\}$ and
\be
E(\rho,\sigma)
\defeq \lim_{\eps\ra 0} \lim_{n \ra \infty} \frac{-\log\beta_n^\eps}{n}
\label{eq:E-def}\ee

Hiai and Petz~\cite{HP91} proved the following quantum analogue of
\lemref{stein}:
\be \kl{\rho}{\sigma} = E(\rho, \sigma).\label{eq:q-stein}\ee
See also \cite{BR12, Li12} for elegant and elementary proofs.  The
``strong converse'' of \eq{q-stein} was proved by Ogawa and
Nagaoka~\cite{ON00}, and can be thought of as showing that
\eq{q-stein} holds when the limit of $\eps\ra 0$ in \eq{E-def} is replaced by any
fixed $\eps\in (0,1)$.

\subsubsection{Asymptotic composite hypothesis testing}
An important generalization of hypothesis testing is when $\rho$ and
$\sigma$ are chosen from sets $R,S\subseteq \cD(V)$, respectively, and
we need to design our test with knowledge only of $R$ and $S$.
This problem is
known as {\em composite hypothesis testing} and is closely related to
the classical Sanov's theorem. 

One case of particular interest to quantum information is when $\rho \in \cD(A\ot B)$ and $S$
is the set of separable states on $A \ot B$, i.e.,
$S = \Sep(A:B).$
The quantity $\kl{\rho}{\Sep}  \defeq \kl{\rho}{\Sep(A:B)}$ is known as the {\em relative entropy of
  entanglement}~\cite{Vedral97} and has been widely studied as an
entanglement measure (see, e.g., Table I in \cite{BCY11}); note that it
is usually written as $E_R(\rho)$.

One challenge in working with the relative entropy of entanglement is
that $\kl{\rho^{\ot n}}{\Sep}$ will not in general be equal to $n \cdot
\kl{\rho}{\Sep}$, reflecting the fact that $\Sep(A^{\ot n} : B^{\ot n})$ is larger
than the convex hull of $\{\sigma_1\ot \cdots \ot \sigma_n : \sigma_1,\ldots,\sigma_n
\in \Sep(A:B)\}$. Intuitively, $\Sep(A^{\ot n}:B^{\ot n})$ can be thought of as
the set of states on the $2n$ systems $A_1\ldots A_nB_1\ldots B_n$
which are separable across the $A_1\ldots A_n : B_1\ldots B_n$ cut,
but may be entangled arbitrarily among the $A$ systems and among the
$B$ systems.  This is an example of the quantum-information phenomenon
known as the {\em additivity} problem (see, e.g., \cite{undecidable, Smith10}).

\begin{dfn}\label{def:RS-fam}
Let $\famR = (R^1,R^2,\ldots)$, $\famS=(S^1,S^2,\ldots)$, with $R^n,S^n
\subseteq \cD(V^{\ot n})$.  Then the {\em asymptotic relative entropy} of
$\famR$ with respect to $\famS$ is
\be \kl{\famR}{\famS} \defeq \lim_{n\ra \infty}
\inf_{\substack{\rho\in R^n\\ \sigma\in S^n}}
\frac{\kl{\rho}{\sigma}}{n} .
\label{eq:DRS}\ee
We further define
\ba
\alpha_n(\cM) &\defeq \sup_{\rho \in R^n} \tr \left( (I-\cM)\rho \right) \label{eq:alpha-n}\\
\beta_n(\cM) &\defeq \sup_{\sigma\in S^n} \tr \left( \cM\sigma \right) \label{eq:beta-n}\\
\beta_n^\eps &\defeq \inf \{\beta_n(\cM) : \alpha_n(\cM)\leq\eps\} \label{eq:beta-n-eps}\\
E(\famR, \famS) &\defeq \lim_{\eps \ra 0} \lim_{n\ra\infty} 
\frac{-\log\beta_n^\eps}{n} \label{eq:ERS}
\ea
\end{dfn}
In Eqs.~\eq{alpha-n} and \eq{beta-n}, we have $\cM\in E(V)$ and in
\eq{beta-n-eps} there is an implicit dependence on $R^n,S^n$.
Note that the limits of Eq. \eq{DRS} (resp. Eq. \eq{ERS}) may
not exist, in which case we leave $D(\famR\dmid\famS)$ (resp. $E(\famR,\famS)$)
undefined.  
See \cite{BowenD06} for a discussion of
replacing the $\lim$ with $\liminf$ or $\limsup$.

An important special case of Eq. \eq{DRS} is the {\em regularized
relative entropy of entanglement}~\cite{VP98}, which is defined to be
$\lim_{n\ra \infty}\frac{1}{n} D(\rho^{\ot n} \dmid \Sep)$, and is normally denoted
$E_R^\infty(\rho)$.  In our notation this quantity is given by
\be \kl{\famrho}{\famSep}.\ee

An important result about composite quantum hypothesis testing is that error exponent
$\min_{\rho\in R^1}D(\rho\|\sigma)$ can be achieved by a test that depends only on $R_1$
and $\sigma$~\cite{qSanov,Hayashi:02d}.  In terms of \defref{RS-fam}, this can be expressed as
\be \kl{\famR}{\famS} = E(\famR , \famS), \label{eq:q-stein-generic}\ee
whenever $\famR,\famS$ are of the form $R^n = \{\rho^{\ot n} : \rho \in
R^1\}$ and $S^n = \{\sigma^{\ot n}\}$, for some set $R^1$ and some
state $\sigma$.
We call results of the form \eq{q-stein-generic} ``quantum Stein's
Lemmas,'' because, like the classical Chernoff-Stein Lemma, they give
an equality between a relative entropy and an error exponent for
hypothesis testing.

A quantum Stein's Lemma has also been proven in the case when
$\famR=\famrho$ for a fixed state $\rho$ and $\famS$ is a family of sets.  In this
case, \eq{q-stein-generic} is proved in~\cite{BP10} in the case where $\famS$ is a
{\em self-consistent} family of states, defined as follows.
\begin{dfn}[\cite{BP10}]\label{def:self-consistent}
$\famS=(S^1, S^2,\ldots)$ is a self-consistent family of states if \benum
\item Each $S^n$ is convex and closed.
\item There exists a full-rank state $\sigma$ such that each $S^n$
  contains $\sigma^{\ot n}$.
\item For each $\sigma\in S^n$, $\tr_n \sigma \in S^{n-1}$.
\item If $\sigma_n\in S^n, \sigma_m \in S^m$ then $\sigma_n \ot
  \sigma_m \in S^{n+m}$.
\item $S^n$ is closed under permutation.
  \eenum
\end{dfn}
Some important cases of self-consistent families of states are
$\famSep$ (defined in \secref{q-stein-intro}),
$\mathsf{PPT}$ (defined in \appref{background}, although it will not
be used in this paper) and $\famsigma$ for any full-rank state $\sigma$.

\subsubsection{Hypothesis testing with restricted measurements}

We now introduce the problem of quantum hypothesis testing with
restricted measurements.  The full set of [two-outcome]
measurements on $V^{\ot n}$ (i.e.~$E_2(V^{\ot n})$) consists of all $\{\cM, I-\cM\}$ where $0\leq \cM
\leq I$.  However, it is often useful to consider smaller classes of
measurements, such as those that two parties can perform with local
operations and classical communication (LOCC). When considering
restricted classes of measurements, our objective might be to minimize
the probability of error (subject to the usual tradeoff between type I
and type II errors), or it might be to maximize the classical relative
entropy of the output distributions.  In the former case we will use
measurements in $E_2(V)$ and in the latter we will use measurements in $E_{\bbN}(V)$.

\begin{dfn}\label{def:RSM-fam}
Let $\famR = (R^1,R^2,\ldots)$, $\famS=(S^1,S^2,\ldots)$, with $R^n,S^n
\subseteq \cD(V^{\ot n})$, and $\famM = (M^1, M^2, \ldots)$, with
$M^n\subseteq E_{\bbN}(V^{\ot n})$. Then the {\em asymptotic relative entropy} of
$\famR$ with respect to $\famS$ under measurements $\famM$ is
\begin{subequations}\ba D_{\famM}(\famR \dmid \famS)& \defeq \lim_{n\ra \infty}
D_{M^n}(R^n\dmid S^n) \\
D_{M^n}(R^n\dmid S^n) & \defeq 
\sup_{\cM \in M^n} \inf_{\substack{\rho\in R^n\\ \sigma\in S^n}}
\frac{D(\cM\left(\rho \right) \dmid \cM\left( \sigma \right))}{n} .
\ea \label{eq:DRSM-family}\end{subequations}
For $\cM\in E(V^{\ot n})$, we further define
\ba
\alpha_n(\cM) &\defeq \sup_{\rho \in R^n} \tr \left( (I-\cM)\rho \right) \\
\beta_n(\cM) &\defeq \sup_{\sigma\in S^n} \tr \left( \cM\sigma \right) .
\ea
Now we restrict $M^n$ to two-outcome measurements and use $\cM$ as a
shorthand for $\{\cM,I-\cM\}$ to define
\ba 
\beta_n^\eps(\famM) &\defeq \inf_{\cM \in M^n \cap E_2(V^{\ot n})} \{\beta_n(\cM) : \alpha(\cM)\leq\eps\}
\\
E_{\famM}(\famR, \famS) &\defeq \lim_{\eps \ra 0} \lim_{n\ra\infty}
\frac{-\log\beta_n^\eps}{n} \label{eq:ERSM-family}
\ea
As before, the quantities \eqref{eq:DRSM-family} and \eqref{eq:ERSM-family}
are left undefined when the corresponding limit does not exist.
\end{dfn}

Following our notation for families of states, we use boldface
(e.g. $\famM$) to denote families of measurements.  In particular, we
define $\SEP(A:B)$ to denote separable measurements on $A:B$ (i.e. $M$
where every POVM element has the form $\sum_i X_i \ot Y_i$ with $X_i, Y_i\geq 0$) and denote the
corresponding family by
$$\famSEP(A:B) = (\SEP(A:B), \SEP(A^{\ot 2}: B^{\ot 2}), \ldots).$$
Again we will often write $\SEP$ or $\famSEP$ where the systems $A,B$
are clear from context.  Note that $\Sep(A:B)$ and $\SEP(A:B)$ both
refer to sets of matrices that can be written as $\sum_i X_i \ot Y_i$
with $X_i,Y_i\geq 0$; the difference is that $\Sep$ refers to density
matrices (i.e. matrices with trace one) and $\SEP$ to measurements
made up from POVM elements (i.e. matrices with operator norm $\leq 1$).

Another important class of measurements is $\ALL^n$, which is simply the
set of all valid quantum measurements on $n$ systems: i.e. $\ALL^n =
E_{\bbN}(V^{\ot n})$.  
The corresponding family is denoted $\famALL$.
Some useful structural facts about $D_{\ALL^n}$ 
are proved in \cite{BertaFT17}.

One further definition we will need (following \cite{Piani09}, but
with different notation) is the idea of a {\em compatible pair}.
\begin{dfn}
If $\famM$ is a collection of measurements and $\famS$
is a collection of states, we say that
$(\famM,\famS)$ are a compatible pair if (a) $\famS$ is closed under permutations of the
systems and under convex combinations, and (b) applying a measurement in $\famM$ to a state in $\famS$ and
conditioning on any outcome leaves a residual state that is still in
$\famS$.   More concretely for positive integers $n,k$, for $\rho_{n+k}\in S^{n+k}$,
for $\cM^k = (\cM_j^k)_{j=1,2,\ldots} \in M^k$, and for $j$ a positive integer, define
$$\tilde\omega_n  = \tr_{n+1,\ldots,n+k}[\rho_{n+k}(I_n \ot \cM_k^j)],$$
and (assuming that $\tr\tilde\omega_n\neq 0$) we define
$$\omega_n  = \frac{\tilde\omega_n}{\tr \tilde\omega_n}.$$
(Here the permutation symmetry of $\famS$ means that we can assume for convenience that
$\cM^k$ acts on the last $k$ systems.)
If $(\famM,\famS)$ is a compatible pair then for any choice of $n,k,j,\rho_{n+k},M_k$,
either $\tr\tilde\omega_n=0$ or $\omega_n \in S^n$.
\end{dfn}

The main example of compatible pair which motivates our work is $(\famSEP, \famSep)$.  We
could also consider $(\LOCC,\famSep)$, or $(\famM,\famSep)$ where $\famM$ is any other
subset of $\famSEP$.   Compatible pairs also arise from resource theories, in which there
is typically a family of free quantum operations and free quantum states, with the
property that the free operations preserve the set of free quantum states. In some cases,
these can be defined by starting with the set of operations (e.g. LOCC operations which
yield the set of separable states) or the set of states (e.g. thermal states of some fixed
Hamiltonians).  We will be interested in a slightly different setting in which quantum
operations are replaced by measurements.  Besides $(\famM,\famSep)$ with
$\famM\subset\famSep$ other examples of compatible pairs are:
\bit
\item {\em Symmetry constraints.} For each $n$, fix a group $G_n$ of
  unitaries acting on $V^{\ot n}$.  These should be compatible in the sense that $G_n
  \ot I \subseteq G_{n+1}$ and $\pi(G_n)=G_n$ for any permutation $\pi$ of the
  $n$ systems.  If $S^n$ is the set of all states that commute with $G_n$ and
  $M^n$ is any subset of the measurements that commute with $G_n$, then 
  $\famS=(S^n)_{n\geq 1}$ and $\famM=(M^n)_{n\geq 1}$ are compatible pairs. This has
  been studied in the context of the resource theory of
  asymmetry~\cite{GS08}.
\item In quantum optics we can take $\famS$ to be the convex hull of Gaussian
  quantum states and 
  $\famM$ the measurements that can be implemented with Gaussian quantum
  operations~\cite{LRWNWA18}.
\item Let $\famS$ be the set of stabilizer states and $\famM$ the set
  of Pauli measurements.  The famous Gottesman-Knill theorem~\cite{Got98} includes the
  fact that $\famS$ is closed under measurements from $\famM$.
\eit
For each of these compatible pairs, if we consider $\famS$ to be set
of free states then the relative entropy $D(\famrho \dmid
\famS)$ can be viewed as a cost of the state $\rho$, with a meaning made more precise in~\cite{HOH02, BG15}.

We will need some more mild regularity conditions on the classes of measurements we
consider.  
\begin{dfn}\label{def:nice-meas}
$\famM = (M^1, M^2, \ldots)$ is a {\em self-consistent} family of measurements if 
\bit
\item For any $k,l$ and any $\cM^k\in M^k, \cM^l\in M^l$, we have $\cM^k\ot \cM^l\in M^{k+l}$ and
  $\cM^k \ot I_l \in M^{k+l}$.
\item $M^n$ is closed under permutations of the $n$ systems.
\item $M^n$ is closed under finite labelled mixtures.  In other words, if
  $\{\cM^{(i)}\}_i$ are a collection of measurements in $M^n$ where $\cM^{(i)}$ has POVM elements
  $\{\cM^{(i)}_j\}_j$ and $\{p_i\}_i$ is a probability distribution then the measurement
  with POVM elements $\{p_i \cM^{(i)}_j\}_{i,j}$ is in $M^n$.  
\eit
\end{dfn}

This last condition on measurements needs a little more explanation.   First, the
measurement outcomes are labelled by pairs of integers, so we need to relax our definition
of $E_{\bbN}(V)$ and allow measurements indexed by any finite set.  Second, observe that
the property of closure under finite labelled mixtures
is implied by the following natural two conditions: (1) that
$M^n$ is convex, and (2) that $M^n$ is closed under relabeling of outcomes, i.e.~if
$(\cM_1,\cM_2,\ldots)\in M^n$ and $\pi:\bbN\mapsto\bbN$ is an injective map then
$(\cM_{\pi(1)}, \cM_{\pi(2)},\ldots)\in M^n$.  These in turn (along with the other
self-consistency properties) are satisfied by all the examples
of families of measurements mentioned in this paper.

Our main results (in Sections \ref{sec:q-stein} and \ref{sec:q-chernoff}) involve
compatible pairs with self-consistent families of measurements, and we also discuss previously known results about compatible pairs in
\secref{superadditivity}.

\subsubsection{Relations between distinguishability measures}
Finally, we state some known and new results that relate the different
versions of $D,E, D_{\famM}, E_{\famM}$.
The following statement is a consequence of the minimax theorem.
\def\LemMinimax#1{Let $R^{#1}$, $S^{#1}\subseteq D(V)$ be closed and
  convex, while $M\subset E_{\bbN}(V)$ is closed under finite labelled mixtures (as
  defined in \defref{nice-meas}).   
 Then
\ba
\sup_{\cM \in M^{#1}} & \min_{\substack{\rho\in R^{#1}\\ \sigma\in S^{#1}}}
D(\cM\left(\rho \right) \dmid \cM\left( \sigma \right))
\twocol{\nn \\ & } = \min_{\substack{\rho\in R^{#1}\\ \sigma\in S^{#1}}}  \sup_{\cM \in M^{#1}}
D(\cM\left(\rho \right) \dmid \cM\left( \sigma \right))
\label{eq:D-minimax}
\ea
}

\begin{lem}\label{lem:D-minimax}
\LemMinimax{}
\end{lem}

Note that the LHS is trivially $\leq$ the RHS, and that the RHS is the form of
restricted-measurement distinguishability introduced by Piani~\cite{Piani09}.

Our Lemma will rely on a minimax theorem that is similar to the minimax
theorems of Kneser, Fan and Sion from the 1950s~\cite{Sion58} but which needs to handle
the possibility that the relative entropy can be infinite.

\begin{lemma}[Thm 5.2 of \cite{FarkasR06}]\label{lem:minimax}
Let $X$ be a compact and convex subset of a Hausdorff topological
vector space and let $Y$ be a convex subset of a linear space.  Let $f:X\times Y \rightarrow
\bbR\cup \{+\infty\}$ be lower semi-continuous on $X$ for fixed $y\in Y$, convex in $x$
and concave in $y$.  Then
\be \sup_{y\in Y}\inf_{x\in X} f(x,y) = \inf_{x\in X}\sup_{y\in Y}f(x,y).\ee
\end{lemma}


\begin{proof}[Proof of \lemref{D-minimax}]
We will take our set $X$ to be $R\times S$ with an element $x$ representing a pair
of density matrices $(\rho,\sigma)$.  Let $\cP(M)$ denote the set of probability
distributions over $M$ with countable support and define $Y=\cP(M)$.  We can now define
\be f((\rho,\sigma),\mu) := 
\E_{\cM \sim \mu}  D(\cM \left(\rho \right) \dmid \cM\left( \sigma \right)).\ee
Clearly $f$ is affine, and hence concave, in $\mu$.  For fixed $\cM$ (and thus fixed $\mu$),
the relative entropy is known to be convex and lower semicontinuous~\cite{Posner75,Donald86}.  Thus
we meet the conditions of \lemref{minimax}. Note also that the lower semicontinuity of $f$
and the compactness of $R\times S$ guarantees that the $\min$ is
achieved.  \lemref{minimax} then implies that
\be 
\min_{\substack{\rho\in R\\ \sigma\in S}}\sup_{\mu \in\cP(M)}
f((\rho,\sigma),\mu)
\leq
\sup_{\mu \in\cP(M)} \min_{\substack{\rho\in R\\ \sigma\in S}}
f((\rho,\sigma),\mu)\label{eq:applied-minimax}\ee
(In fact it establishes an equality but we write $\leq$ to emphasize
the direction that we are trying to prove.)

Eq.~\eq{applied-minimax} is close to what we want but has $\cP(M)$ in place of $M$.
Since $\cP(M)$ includes distributions which assign probability 1 to
a particular measurement, we have
\ba 
\min_{\substack{\rho\in R\\ \sigma\in S}}\sup_{\cM\in M}
& D(\cM\left(\rho \right) \dmid \cM\left( \sigma \right))
\twocol{\nn \\ & }\leq
\min_{\substack{\rho\in R\\ \sigma\in S}}\sup_{\mu \in\cP(M)}
f((\rho,\sigma),\mu)
\label{eq:M-to-PM}.\ea

Upper bounding the $\sup$ over $\cP(M)$ in terms of a $\sup$ over
$\cM$ is less trivial, and will need to use the fact that $M$ is
closed under countable labeled mixtures.  
  Fix $\rho,\sigma,\mu$ and suppose that
$\mu$ assigns probability $p_i$ to $\cM^{(i)}$ for $i=1,2,\ldots$.  Let
$\{\cM^{(i)}_j\}_{j=1,2,\ldots}$ be the POVM elements of $\cM^{(i)}$.  Then we will define
the measurement $\cM$ with POVM elements $\{p_i \cM^{(i)}_j\}_{i,j}$, and by our
hypothesis, $\cM\in M$.
Then
\ba
f&((\rho,\sigma),\mu) = \sum_i p_i D(\cM^{(i)}(\rho) \dmid
\cM^{(i)}(\sigma)) \nn \\
& = \sum_{i,j} p_i \tr [\cM^{(i)}_j\rho] 
(\log \tr [\cM^{(i)}_j\rho]  - \log \tr [\cM^{(i)}_j\sigma] )\nn \\
& = \sum_{i,j} \tr [p_i\cM^{(i)}_j\rho] 
(\log \tr [p_i\cM^{(i)}_j\rho]  - \log \tr [p_i\cM^{(i)}_j\sigma] )\nn
\\
& = D(\cM(\rho) \dmid \cM(\sigma))
\ea 
We can take the minimum over $\rho,\sigma$ to obtain
\be \min_{\substack{\rho\in R\\ \sigma\in S}} f((\rho,\sigma),\mu)
\leq
\min_{\substack{\rho\in R\\ \sigma\in S}} 
D(\cM\left(\rho \right) \dmid \cM\left( \sigma \right)), \ee
where $\cM$ depends on $\mu$.  Next we can take the $\sup$ over $\mu$
to obtain
\be \sup_{\mu \in\cP(M)}\min_{\substack{\rho\in R\\ \sigma\in S}} f((\rho,\sigma),\mu)
\leq
\sup_{\cM \in M}\min_{\substack{\rho\in R\\ \sigma\in S}} 
D(\cM\left(\rho \right) \dmid \cM\left( \sigma \right)).\label{eq:PM-to-M} \ee

Finally combining the inequalities \eq{M-to-PM}, \eq{applied-minimax}
and \eq{PM-to-M} implies the proof of the lemma.

\end{proof}

We remark  that some versions of the minimax theorem (i.e. Thm 4.2 of \cite{Sion58})
require only a weaker form of concavity in which for any $p\in [0,1]$ and any $x\in X, y_1,y_2\in
Y$, there exist $y_0\in Y$ such that $f(x,y_0) \geq p f(x,y_1) + (1-p) f(x,y_2)$.  In
other words, $y_0$ does not have to be $p y_1 + (1-p)y_2$ but could be an arbitrary point
and indeed $Y$ does not even have to be a linear space. This would perfectly fit our
approach of taking labelled mixtures of measurements.  However, since our theorem needs to
handle the possibility that $D(\cdot\dmid\cdot)=\infty$, we cannot directly use Thm 4.2 of
\cite{Sion58}.

\medskip
\noindent {\em Known facts:} The following relations between the quantities have been derived previously.
\ba
\twocol{\intertext{quantum Stein's Lemma~\cite{HP91}:}}
E(\famrho, \famsigma)  & = D(\rho \dmid\sigma)
\onecol{& \text{quantum Stein's Lemma~\cite{HP91}}}  \\
\twocol{\intertext{If $\famS$ satisfies property (4) of \defref{self-consistent}:}}
D( \{\rho \} \dmid S^1 ) & \geq D( \famrho \dmid \famS  )
\onecol{& \text{if $\famS$ satisfies property (4) of \defref{self-consistent}}}
\\
\twocol{\intertext{From monotonicity of relative entropy:}}
D(\famR \dmid \famS) & \geq D_{\famM}(\famR \dmid \famS)
\onecol{& \text{from monotonicity of relative entropy} }
\label{eq:M-monotonicity}\\
\twocol{\intertext{If $\famS$ is a self-consistent family (\defref{self-consistent})\cite{BP10}:} }
E( \famrho, \famS ) & = D( \famrho \dmid \famS ) \label{eq:BP-equality}
\onecol{& \text{for $\famS$ a self-consistent family (\defref{self-consistent})\cite{BP10}} }
\ea

We can, in fact, relate $D_{\ALL}, D, E$ for any $\rho$ and any closed convex $\famS$ using
\be D_{\ALL}(\famrho \dmid \famS)
\stackrel{\eq{limsup-beta}}\geq
E(\famrho, \famS)
\stackrel{\eq{BP-equality}}{=}
D(\famrho \dmid \famS)
\stackrel{\eq{M-monotonicity}}{\geq}
D_{\ALL}(\famrho \dmid \famS) \ee
The main goal of the second half of this paper is to extend these results as far
as possible to $D_{\famM}$ and $E_{\famM}$.

\subsubsection{Superadditivity}\label{sec:superadditivity}
When we consider families of states and measurements, it is not {\em a
  priori} clear whether the distinguishability per system should
increase or decrease with the number of systems.  We say that a
quantity $f(\rho)$ is {\em subadditive} if $f(\rho_{XY}) \leq
f(\rho_X) + f(\rho_Y)$ (e.g., entropy) and {\em superadditive} if
$f(\rho_{XY}) \geq f(\rho_X) + f(\rho_Y)$ (e.g., most entanglement
measures).  A function $f$ is weakly subadditive
$f(\rho^{\ot n})\leq n f(\rho)$ and is weakly superadditive
if $f(\rho^{\ot n})\geq n f(\rho)$).
If a function is both (weakly) subadditive and superadditive then we
say it is (weakly) {\em additive.}

One of the main results known so far about relative entropy with
restricted measurements is due to Piani~\cite{Piani09}, who used these
measures to prove a superadditivity inequality.
\ba
\twocol{\intertext{For compatible $(\famM, S)$ \cite{Piani09}:}}
D (\rho_{XY} \dmid S^2) & \geq D_{\famM}(\rho_X \dmid S^1) + D(\rho_Y \dmid S^1)
\onecol{& \text{for compatible $(\famM, S)$ \cite{Piani09}}}
\label{eq:Piani-asym} \\
\twocol{\intertext{As a corollary of \eq{Piani-asym} \cite{Piani09}:}}
D(\famrho \dmid \famS) & \geq D_{\famM}(\rho \dmid S^1)
\onecol{& \text{as a corollary of \eq{Piani-asym} \cite{Piani09}}}
  \ea

In fact, Piani's result can easily be improved to show that
$D_{\famM}(\famR\dmid\famS)$ is superadditive whenever $(\famM,\famR)$ and
$(\famM,\famS)$ are compatible pairs, or in fact when $\famR$ satisfies a milder condition.

\def\LemSuperadd{
Let $(\famM, \famS)$ be a compatible pair with $\famM$ a
self-consistent family. Let $\famR$ be a family of states that is closed under partial
trace, i.e.~satisfying $\tr_jR^n \subseteq R^{n-1}$ for each $1\leq j\leq n$.
Then for all $\rho_{XY}\in \cD(V^{\otimes k} \otimes V^{\otimes l})$, if we identify $X$
with $V^{\otimes k}$ and $Y$ with $V^{\ot l}$, we have
\be \label{eq:superaddidivitymarco}
D_{M^{k+l}}( \rho_{XY} \dmid S^{k+l} ) \geq D_{M^k}( \rho_{X} \dmid S^k) 
+ D_{M^l}( \rho_{Y} \dmid S^l ).
\ee
Moreover,
\ba\label{eq:superaddidivitymarco2}
D_{\famM}(\famR \dmid \famS)
& = \lim_{n\ra\infty}\frac 1n D_{M_n}(R^n \dmid S^n)
\twocol{\nn \\ &} 
= \sup_n \frac 1n D_{M_n}(R^n \dmid S^n).
\ea
}
\begin{lem}\label{lem:superadd}
\LemSuperadd
\end{lem}

\begin{proof}
The argument is a direct adaptation of the proof of Theorem 1 in \cite{Piani09}.


 Let $\cM^X\in M^k, \cM^Y \in M^l$ be arbitrary.
Define an orthonormal basis $\ket{1},\ket{2},\ldots$ corresponding to the outcomes
$1,2,\ldots$ of $\cM^X$.
Define $p_i(\rho_X) = \tr(\cM_i^X
\rho_X) = (\cM^X(\rho))_i$ and $\rho_Y^i = \tr_{X}[(\cM_i^X \otimes I_Y) \rho_{XY}] /
p_i(\rho_X)$.  Choose $\sigma_{XY} \in S^{k+l}$ and define $p_i(\sigma_X)$
and $\sigma_Y^i$ analogously.
\twocol{Then in \eq{marcochain} (at the top of the next page) we derive a lower bound for 
  $D((\cM^X \ot \cM^Y)(\rho_{XY}) \dmid (\cM^X \ot \cM^Y)(\sigma_{XY}) )$.
\begin{figure*}[!t]
 \normalsize
\setcounter{MYtempeqncnt}{\value{equation}}
}
\begin{subequations}\label{eq:marcochain}
\ba
D&\left((\cM^X \ot \cM^Y)(\rho_{XY}) \dmid (\cM^X \ot \cM^Y)(\sigma_{XY}) \right) 
\nonumber \\  &= D \left(\sum_{i\geq 1} p_i(\rho_X) \proj i \ot \cM^Y(\rho_Y^i) \dmid
  \sum_{i\geq 1} p_i(\sigma_X) \proj{i} \ot \cM^Y(\sigma_Y^i) \right) \label{eq:marco1} \\
\intertext{\hfill follows from Proposition 1 of  \cite{Piani09}}
&=  D\left( \cM^X(\rho_X) \dmid \cM^X(\sigma_X) \right) + \sum_{i\geq 1} p_i(\rho_X) D \left(\cM^Y(\rho_Y^i) \dmid \cM^Y(\sigma_Y^i) \right) \label{eq:marco2} \\
&\geq D \left( \cM^X(\rho_X) \dmid \cM^X(\sigma_X) \right) + D\left (\sum_{i\geq 1}
  p_i(\rho_X) \cM^Y(\rho_Y^i) \dmid \sum_{i\geq 1} p_i(\rho_X) \cM^Y(\sigma_Y^i)     \right) \label{eq:marco3} \\
\intertext{\hfill from joint convexity of relative entropy}
&= D \left( \cM^X(\rho_X) \dmid \cM^X(\sigma_X) \right) + D\left(
  \cM^Y(\rho_Y) \dmid \cM^Y\left(\sum_{i\geq 1}  p_i(\rho_X) \sigma_Y^i \right)   \right),
\label{eq:marco4}
\\
\intertext{\hfill from linearity of the measurement}& \geq 
\inf_{\tilde\sigma_X\in S^k}
D\left(\cM^X(\rho_X) \dmid \cM^X(\tilde\sigma_X)\right) + 
\inf_{\tilde\sigma_Y\in S^l}
D\left(\cM^Y(\rho_Y) \dmid \cM^Y(\tilde\sigma_Y)\right)
\label{eq:marco5}
\ea
\end{subequations}
\twoone{\setcounter{equation}{\value{MYtempeqncnt}}
In Piani's proof in \cite{Piani09} the
analogues of \eq{marco3}, \eq{marco4} were Lemma 1 and Property 2 of
Proposition 1 respectively.

\hrulefill
\vspace*{4pt}
\end{figure*}}
{where \eq{marco1} follows from Proposition 1 of 
\cite{Piani09}, \eq{marco2} from direct calculation, \eq{marco3} from joint convexity of
relative entropy, and \eq{marco4}  
from linearity of the measurement.  (In Piani's proof in \cite{Piani09} the
analogues of the third and fourth lines were Lemma 1 and Property 2 of
Proposition 1 respectively.)}

We now take the infimum over $\sigma_{XY}$ and then the supremum over $\cM^X,\cM^Y$,
yielding \eqref{eq:superaddidivitymarco}.    If we instead take the infimum over
$\sigma_{XY}$ and $\rho_{XY}$ before taking the supremum over $\cM^X,\cM^Y$, then we find
that $D_{M^{k+l}}(R^{k+l} \dmid S^{k+l}) \geq D_{M^{k}}(R^{k} \dmid S^{k})
 + D_{M^{l}}(R^{l} \dmid S^{l})$.  In other words, if $f(n) = D_{M_n}(R^n \dmid S^n)$ then
$f$ is superadditive (i.e.~$f(k+l)\geq f(k)+f(l)$). This implies
(\ref{eq:superaddidivitymarco2}).

\end{proof}

The preceding lemma says that $D_{\famM}(\rho \dmid \famS)$ is a superadditive function of
$\rho$ for
compatible pairs $(\famM, \famS)$.  The compatibility requirement here is essential.  The
pair $(\famALL,\famSep)$ is not compatible, and $D(\cdot \dmid \famSep)$ is known to be
strictly subadditive (i.e.~not superadditive) in some cases~\cite{VW01}.  
This does not directly
yield an example of strict subadditivity for $D_{\famALL}(\cdot\dmid\famSep)$ but can be
modified to do so.  The example in \cite{VW01} is the antisymmetric Werner state $\rho = \frac{I
- \swap}{d(d-1)} \in \cD(\bbC^d \ot \bbC^d)$. In \cite{VW01}, it is proved that 
\bas D(\rho \| \Sep(\bbC^d : \bbC^d)) & = 1
\twoone{\\}{\qquad\text{and}\qquad}
D(\rho\ot\rho \| \Sep(\bbC^{d^2} : \bbC^{d^2})) \twocol{&} = 1 + O(1/d),\eas
showing that $D(\cdot \dmid \famSep)$ can be strictly subadditive.  Observe that if we
measure $\rho$ with the two outcome measurement $\{\frac{I \pm \swap}{2}\}$ and label the
outcomes +/- then we will always obtain the outcome - while for any $\sigma\in\Sep$ we
have $\Pr[-] \leq 1/2$.  Thus $D_{\ALL}(\rho \| \Sep(\bbC^d:\bbC^d)) \geq 1$ (and in fact
equality holds).  On the other
hand, monotonicity of relative entropy implies that
\ba D_{\ALL}(\rho\ot \rho \| \Sep(\bbC^{d^2} : \bbC^{d^2})) 
& \twocol{\nn\\}\leq 
D(\rho\ot\rho \| \Sep(\bbC^{d^2} : \bbC^{d^2})) & = 1 + O(1/d).
\ea
Thus we have an example where $D_{\famALL}(\cdot \| \famSep)$ is strictly subadditive.

On the other hand, $D_{\famM}(\cdot \dmid \famS)$ can be strictly superadditive (i.e., not
subadditive).  Let us consider the simple situation in which $R^n = \{\rho^{\ot n}\}$ and
$S^n = \{\sigma^{\ot n}\}$.  It is a consequence of the quantum Stein's Lemma \eq{q-stein}
(see also \cite{Hayashi:02d}) that
$$D(\rho\dmid\sigma) = \lim_{n\ra \infty} \frac{1}{n} D_{\ALL}(\rho^{\ot
  n} \dmid \sigma^{\ot n}).$$
Thus, any example in which
\be \max_{M\in\ALL} D(M(\rho) \dmid M(\sigma)) < 
D(\rho\dmid\sigma)\label{eq:strict-super}\ee
will yield an example in which $D_{\famM}(\cdot \dmid \famS)$ is strictly superadditive.
In fact, Proposition 5 of \cite{BertaFT17} (building upon Lemma 1 of
\cite{Petz-monotonicity}) states that \eq{strict-super} holds whenever 
$D(\rho\dmid\sigma)$ is finite and $\rho\sigma\neq \sigma\rho$.  Thus superadditivity is a
generic property of $D_{\famM}(\cdot \dmid \cdot)$.

\subsection{A quantum Stein's Lemma for restricted measurements}
\label{sec:q-stein}


\begin{thm}[Quantum Stein's Lemma for restricted measurements] \label{thm:q-stein}
For any compatible pairs $(\famM, \famR)$ and $(\famM, \famS)$ with $\famM$ a
self-consistent family and $\famR,\famS$ closed,
\be
D_{\famM}(\famR \dmid \famS) = E_{\famM}(\famR ,\famS)\,.
\ee
\end{thm}

\begin{proof}

For any positive integer $k$, suppose  $0 \leq E_k<  \frac 1 k D_{M^k}(R^k \dmid
S^k)$. (If $D_{M^k}(R^k \dmid S^k) =\infty$ then this means that $E_k$ is an arbitrary
nonnegative number.)  The supremum over measurements in the definition of $D_{M^k}(R^k
\dmid S^k)$ means that there exists $\cM^k\in M^k$ such that
$$ \frac 1k D(\cM^k(R^k) \dmid \cM^k(S^k))> E_k.$$
Define  $P \defeq \cM^k(R^k)$ and $Q \defeq  \cM^k(S^k)$. Then
\be \frac 1k D(p\dmid q) >E_k \qquad\forall p\in P, q\in Q.
\label{eq:M-def}\ee
Given a state $\rho\in\cD(V^{\ot nk})$, we apply $\cM^k$ to each block of
$k$ systems, obtaining outcomes $x_1,\ldots,x_n$.
Then since $(\famM, \famR)$ and $(\famM, \famS)$
are compatible pairs, the distribution of each $x_i$, conditioned on
any possible value of $x_1,\ldots,x_{i-1}$, is an element of
$P$  (if $\rho\in R^{nk}$) or $Q$ (if $\rho\in S^{nk}$).  Thus,
according to \thmref{adaptive}, there is a sequence of acceptance regions that
achieves the rate $E_k$.  Thus for any $\eps\in (0,1)$, 
\be
\liminf_{n\ra\infty} -\frac{1}{nk} \log \beta_{nk}^\eps
\geq E_k
\label{eq:k-copy-achievable}\ee
Given a state in $R^{nk+l}$ or $S^{nk+l}$ for $l<k$ we can discard $l$ systems and obtain
a state in  in $R^{nk}$ or $S^{nk}$, using the fact that $\famR,\famS$ are closed under
partial trace.  Thus we can drop the $k$-dependence from the LHS of \eq{k-copy-achievable}
to obtain
\be
\liminf_{n\ra\infty} -\frac{1}{n} \log \beta_{n}^\eps
\geq E_k.
\ee
Since this holds for any $k$, we can take the $\limsup$ over $k$ to find
\be
\liminf_{n\ra\infty} -\frac{1}{n} \log \beta_{n}^\eps
\geq 
\limsup_{k\ra\infty} \frac 1k D_{M^k}(R^k \dmid S^k) = D_{\famM}(\famR\dmid \famS).
\label{eq:liminf-beta}\ee
This last equality is due to \lemref{superadd}.  

The reverse inequality can be obtained by the following standard argument which we adapt
from \cite{HP91}.   See also footnote 11 of \cite{BCY11} where roughly the same
result was stated and attributed to \cite{HP91,ON00}.  We include a proof here
for completeness and because previous work did not technically show the same results.

For a positive integer $n$ and $\eps>0$, let $\cM := (\cM,I-\cM) \in M_n$ be a measurement such that
$\tr [\cM\rho] > 1-\eps$.
Then  for any $\rho\in R_n,\sigma\in  S_n$,
\ba
\sup_{\cM'\in M_n} & D(\cM'(\rho) \dmid \cM'(\sigma))
\twoone{\nn \\ \geq & }{\geq}  D(\cM(\rho) \| \cM(\sigma)) \nn\\
= & \tr[\cM\rho]\log \tr [\cM\rho] - \tr[\cM\rho]\log \tr[\cM\sigma]\nn \\
& + \tr[(I-\cM)\rho]\log \tr [(I-\cM)\rho]
\twocol{\nn \\ & }- \tr[(I-\cM)\rho]\log \tr[(I-\cM)\sigma] \nn\\
 \geq & -h_2(\tr[\cM\rho]) -  \tr[\cM\rho]\log \tr[\cM\sigma] \nn\\
 \geq & -1 -  (1-\eps)\log \tr[\cM\sigma] 
\ea
Here we define $h_2(p) = -\log(p)-\log(1-p)$ and take log to be base 2.
Rearranging yields
\be -\log \tr[\cM\sigma] \leq \frac{1 + \sup_{\cM'\in M_n} D(\cM'(\rho) \dmid
  \cM'(\sigma))}{1-\eps}.\twocol{\nn}\ee
To relate this to $\beta^\eps_n(\famM)$ we take the $\inf$ over
$\rho\in R_n,\sigma\in S_n$ and then the $\sup$ over $\cM\in M^n\cap E_2(V^{\ot n})$
satisfying $\alpha_n(\cM)\leq\eps$.  
This implies that
\be -  \beta^\eps_n(\famM) \leq
\frac{1 + \inf\twocol{\limits}_{\substack{\rho\in R_n\\ \sigma\in S_n}}\sup\twocol{\limits}_{\cM'\in M_n} D(\cM'(\rho) \dmid
  \cM'(\sigma))}{1-\eps}.
\ee
We can now use \lemref{D-minimax} to exchange the $\inf$ and $\sup$.  Finally we can
divide by $n$ and take the $\limsup$ in $n$ to obtain
\be \limsup_{n\ra\infty} -\frac{1}{n}\beta_n^\eps \leq
\frac{1}{1-\eps}  \limsup_{n\ra\infty} \frac 1n D_{M_n}(R_n\dmid S_n)
\label{eq:limsup-beta}\ee

Combining \eq{liminf-beta} and \eq{limsup-beta} and taking $\eps\ra 0$ we finally
 establish that 
\be E_{\famM}(\famR,\famS) = D_{\famM}(\famR\dmid \famS)
\label{eq:DE-monotonicity}.\ee


\end{proof}



This is analogous to the result in \cite{BP10}, which established $E(\famrho, \famS ) =
D_{\famALL}(\famrho \dmid \famS )$ for self-consistent sets of states $\famS$, but incomparable
because in general $(\famALL, \famS)$ will not be a compatible pair.

While this shows that the optimal hypothesis testing rate for this
restricted-measurement setting does indeed reduce to a relative
entropy, it may be difficult to compute $D_{\famM}$ because of the
regularization (i.e. $\lim_{n\ra\infty}$) and optimization over
measurements in \eq{DRSM-family}.  However, in some special cases, it
is known how to carry out this optimization; e.g. \cite{HayashiO17}
computes the relative entropy of a pure entangled state with respect
to the maximally mixed state under various restricted classes of measurements.

\subsection{Stronger Subadditivity of Quantum Entropy } \label{sec:q-SSA}

We now present an application of \thmref{q-stein} to a strengthening of the celebrated strong subadditivity
inequality of Lieb and Ruskai for the quantum entropy \cite{LR73}, which can be written as
\be
I(A:B \mid C)_\rho \geq 0\,
\ee
where
\bas I(A \twocol{&}:B \mid C)_\rho
\twocol{\\} &   \defeq  H(AC) _\rho + H(BC) _\rho - H(ABC) _\rho - H(C) _\rho
\\ &\defeq  H(\rho_{AC}) + H(\rho_{BC}) - H(\rho_{ABC}) - H(\rho_C)\eas
 denotes the conditional
mutual information of a state $\rho_{ABC}$.   In what follows we will
often omit the subscript $\rho$ when the state is understood.  See
\appref{background} for additional discussion.

In \cite{BCY11}, the following lower bound was shown for any
state $\rho_{ABC}$:
\ba
I(A:B\mid C) &\geq D_{\ALL}( \famrho_{ABC} \dmid \famSep(A:BC) )
\twocol{\nn \\ & }-  D_{\ALL}( \famrho_{AC} \dmid \famSep(A:C) )
\ea
Moreover the following inequality was shown
\ba
D_{\ALL}&( \famrho_{ABC} \dmid \famSep(A:BC) ) - D_{\ALL}(
\famrho_{AC} \dmid \famSep(A:C) )
\twocol{\nn \\ & }\geq  E_{\text{1-LOCC}}( \famrho ,\famSep(A:B)),
 \label{eq:monogamy}\ea
with $\text{1-LOCC}$ the class of all measurements that can be
implemented by quantum local operations and classical communication
from Bob to Alice (see \appref{background} for the precise
definition). This implies that the conditional mutual information is
lower bounded by  $E_{\mathbf{\text{1-LOCC}}}(\famrho,\famSep(A:B))$.
(Ref.~\cite{BCY11} actually stated a weaker result in terms of the 1-LOCC
(trace) distance, but their proof essentially contains \eq{monogamy} as an
intermediate step.   In reading \cite{BCY11,LW12} beware that they use the symbols $D$ and
$E$ with meanings reversed from our conventions.)

In \cite{LW12} the following apparent strengthening of \eq{monogamy} was obtained:
\ba
D_{\ALL}&(\famrho_{ABC} \dmid \famSep(A:BC) )\twoone{\nn \\ \geq &}{\geq} D_{\ALL}(\famrho_{AC} \dmid \famSep(A:C) )
\twocol{\nn \\ & }+ D_{\text{1-LOCC}}( \famrho_{AB} \dmid \famSep(A:B))\,,
 \label{eq:monogamy2} \ea
which implies
\be \label{eq:stronger}
I(A:B\mid C) \geq D_{\text{1-LOCC}}( \famrho_{AB} \dmid \famSep(A:B))\,.
\ee
At the time of \cite{LW12} it was known only that $D_{\text{1-LOCC}}\geq
E_{\text{1-LOCC}}$ (see discussion in the proof of \thmref{q-stein}) and so
\eq{monogamy2} was believed to be stronger than \eq{monogamy}.  
\thmref{q-stein} shows that \eq{monogamy2} is equivalent to \eq{monogamy} and so
it can be used in conjunction with \cite{BCY11} to give an alternative proof of
\eq{stronger}.  This possibility was already discussed in \cite{BCY11}; see the
discussion  surrounding Eq.~(43) of that paper.

\subsection{Symmetric hypothesis testing with restricted measurements}
\label{sec:q-chernoff}

Our main result on symmetric hypothesis testing against an adaptive
adversary (\thmref{chernoff}) makes it natural to conjecture a corresponding result for
symmetric quantum hypothesis testing.   For quantum states
$\rho,\sigma$, define
\ba \Gamma^*(\rho,\sigma) &\defeq \max_{0\leq \lambda\leq 1}
\Gamma^\lambda(\rho,\sigma) \twocol{\nn \\ &}
\defeq  \max_{0\leq \lambda\leq 1} -\log \tr(\rho^\lambda
\sigma^{1-\lambda})
\\ \Gamma_{\famM}^*(\famR, \famS) & \defeq \lim_{n\ra \infty}
\sup_{\cM \in M^n} \inf_{\substack{\rho\in R^n\\ \sigma\in S^n}}
\frac{\Gamma^*(\cM\left(\rho \right) \dmid \cM\left( \sigma \right))}{n}
\\ \gamma_{\famM}(\famR, \famS) &
\defeq \lim_{n\ra \infty}
\sup_{\cM \in M^n} \inf_{\substack{\rho\in R^n\\ \sigma\in S^n}}
\twocol{\nn \\ & \qquad}
-\frac{1}{n} \log \tr (\cM\sigma  + (I-\cM)\rho)
\ea

A quantum analogue of  Chernoff's Theorem was proven in \cite{NS06,ACM07}
and in our notation can be expressed as
$$\gamma_{\famALL}(\famrho, \famsigma) =
\Gamma^*(\rho,\sigma).$$
With restricted measurements, we might ask whether an analogue of \thmref{q-stein} holds.
\begin{con}\label{con:q-chernoff}
If $(\famM, \famR)$ and $(\famM, \famS)$ are compatible pairs, then
$$\gamma_{\famM}(\famR, \famS) =
\Gamma^*_{\famM}(\famR, \famS).$$
\end{con}
A plausible route to proving the conjecture is to use the strategy of the proof of
\thmref{q-stein}, replacing the adversarial Chernoff-Stein Lemma with the adversarial
Chernoff's Theorem (\thmref{chernoff})).  However, there are several limits and sup/inf
steps and we have not verified that these compose in the required ways.

\subsection{Open questions}
Having established a quantum Stein's Lemma for restricted
measurements, we would like to know if a strong converse can also be
proven, or more generally if we can calculate the error exponent for
the type-2 error when the type-1 error is required to be $\leq\eps$ for some
fixed $\eps\in (0,1)$.  The difficulty is that $D_{\famM}(\cdot \dmid \famS) >
D_{M^1}(\cdot \dmid S^1)$ in general, and we would need to control the rate of
convergence as a function of $n$ in the $\lim$ used to define
$D_{\famM}(\cdot \dmid \famS)$.

Like many information-theoretic quantities, $D(\famrho \dmid \famSep)$
and $D_{\famM}(\famrho \dmid \famSep)$ (for various natural choices of
$\famM$) are operationally interesting, but are hard in practice to
compute.  We would like to know the complexity of estimating them
(which is a variant of the usual question about the hardness of
testing separability, cf. \cite{HM10,BCY-stoc}) and whether good
relaxations exist (cf. \cite{BeigiS10}).

Finally, a major application of restricted-measurement
distinguishability is to the related questions of $k$-extendable
states\footnote{A bipartite state $\rho_{AB}$ is said to be
  $k$-extendable if there exists a state $\tilde
\rho_{AB_1\ldots B_k}$ such that $\tilde \rho_{AB_i} = \rho_{AB}$ for
each $i$.  The idea of $k$-extendability was introduced in
\cite{RW89, DohertyPS04}, where it was proved that for any fixed dimension of $A$ and/or $B$, the set
of $k$-extendable states approaches the set of separable states.
However, the rate of convergence is an open question.}, tripartite states with low conditional mutual information
(i.e.~``approximate Markov states'', cf. \cite{ILW08}), and the quality of
approximations achieved by the sum-of-squares hierarchy
(cf. \cite{BHKSZ12}).
  A few of
the more prominent open questions here are:
\bit
\item If $I(A:B\mid E)_\rho$ is small then it was recently discovered~\cite{FR15,Sutter18}
 that an
  ``approximate recovery'' map $T:E\ra E\ot B$ exists such that $(\id \ot
  T)\rho_{AE} \approx \rho_{ABE}$ in the sense of (among other measures) the measured
  relative entropy, i.e.
\be D_{\ALL}(\rho_{ABE} \dmid (\id \ot T)\rho_{AE}) \leq I(A:B\mid E)_\rho.\ee
 Ref.~\cite{Fawzi2} found that we cannot replace the $D_{\ALL}(\cdot\dmid\cdot)$ on the LHS with the ordinary quantum relative
 entropy $D (\cdot\dmid\cdot)$.  Their result leaves open the question of what
 relationship beween $\min_T D(\rho_{ABE}\dmid (\id \ot T)\rho_{AE})$ and
 $I(A:B\mid E)$ is possible.  Can we simply 
 multiply $I(A:B\mid E)$  by some constant, or can these quantities differ by an
 amount that grows with dimensions?  We do not even know whether the ratio between
 these quantities could be arbitrarily large in fixed dimensions.
 
\item How large can $D_{\famM}(\rho \dmid \Sep)$ be when $\rho$ is
  $k$-extendable and $\famM$ is the class of separable measurements?
  Sharp bounds are known~\cite{BCY-stoc} when $\famM=\text{1-LOCC}$,
  and if they could be extended to separable measurements it would
  have implications for quantum Merlin-Arthur games with multiple
  Merlins~\cite{HM10} as well as for classical optimization algorithms.
\item  The ability of semidefinite programming hierarchies to
  estimate small-set expansion can be understood in terms of a
  restricted-measurement distinguishability problem~\cite{BHKSZ12}.
  A major open question is whether small-set expansion on graphs of
  size $n$ can be well-approximated by $O(\log n)$ levels of these
  hierarchies, which would imply a quasipolynomial-time algorithm for
  the problem.  Can tools from quantum information shed further light here?
\eit

\remove{
Can we, like \cite{BP10b}, obtain a reversible resource theory here?  Without a
strong converse, the ``formation'' part will be missing.  From the
achievability of hypothesis testing, we do obtain a distillation map,
but the restriction to measurements in $\famM$ is not very interesting.
It simply means that our distillation map is of the form $\Lambda(X) =
\tr (MX) \Phi^{\ot n D_{\famM, S}^\infty(\rho)} + (1-\tr MX) \tau$, for
some separable $\tau$, and with the restriction that $M\in
\famM$. However, the restriction on $M$ (e.g. if $\famM=\LOCC$) does not
imply that $\Lambda$ can be implemented using comparable resources.
In general we can only say that $\Lambda$ is separability-preserving,
just like in BP2.
}

\appendix

\section{Appendix: Background on Quantum Information}
\label{sec:background}

This appendix contains a very brief review of the quantum formalism
and notation used in this paper.  For a much more detailed introduction to
quantum information theory, see \cite{Wilde-book}, or for an overview
of the field of quantum computing and quantum information more
generally see \cite{NC00, Kitaev:02a}.

\medskip
\noindent
{\bf Density matrices.}
The quantum analogue of a probability distribution over $[d]=
\{1,\ldots,d\}$ is called a {\em density matrix}, or simply a {\em
  state}.  Density matrices must be positive semi-definite and have
trace one.  These conditions are analogous to the requirement that
probabilities must be nonnegative and normalized; indeed diagonal
density matrices correspond exactly to probability distributions.  If
$A$ is a finite-dimensional Hilbert space, then define $\cD(A)$ to be the set of
density matrices on $A$, meaning the set of operators on $A$ that are
positive semi-definite and have trace one.  Let $\cL(A,B)$ denote the
set of bounded linear operators from $A$ to $B$, and let $\cL(A) := \cL(A,A)$.

\medskip
\noindent
{\bf Tensor product.}
To describe composite quantum systems, we use the tensor product.  The
tensor product of a vector  $x\in \bbC^{d_1}$ and a vector
$y\in \bbC^{d_2}$ is denoted $x\ot y$ and has entries that run over
all $x_{i_1}y_{i_2}$ for $i_1\in [d_1], i_2\in [d_2]$.  Similarly, if
$X$ and $Y$ are matrices, then their tensor product $X\ot Y$ has
matrix elements $(X\ot Y)_{(i_1, i_2), (j_1, j_2)} = X_{i_1, j_1}
Y_{i_2, j_2}$.   For vector spaces $A, B$, we let $A\ot B$ denote the
span of $\{a \ot b : a\in A, b\in B\}$.  Note that $\bbC^{d_1} \ot
\bbC^{d_2} \cong \bbC^{d_1 d_2}$.  Finally, in each case we use the
tensor power notation $X^{\ot n}$ to stand for
$$ \overbrace{X \ot X \ot \cdots \ot X}^{n \text{ times}}.$$

\medskip
\noindent
{\bf Product and separable states.}
The tensor product is used to combine quantum states in the same way
that independent classical probability distributions are combined to
form a joint distribution.  Indeed, if $p,q$ are probability
distributions of independent random variables, then $p \ot q$ denotes
the joint distribution.  Similarly, if $\rho$ and $\sigma$ are density
matrices, then $\rho \ot \sigma$ denotes the state of a system that is
in a so-called {\em product state}.     The convex hull of the set of
product states is called the set of {\em separable states}.  We write
$\Sep(A:B)$ to indicate the split along which we demand that the
states be separable, e.g.
\be \Sep(A:B) = \conv\{\alpha \ot \beta : \alpha\in \cD(A),
\beta\in\cD(\beta)\}.\ee
Although the set $\Sep(A:B)$ is convex, it is not easy to work with.
For example, computational hardness results are known for the weak
membership problem.
Instead, it is sometimes more convenient to
consider the relaxation $\PPT$, which denotes the set of states with
Positive Partial Transpose.   The partial transpose operator $\Gamma$
(meant to resemble the right half of the $T$ that usually denotes
transpose) acts linearly on $\cL(A\ot B)$ by mapping $X \ot Y$ to $X
\otimes Y^T$; equivalently we can write it as $\id_A \ot T_B$, where $\id_A$ is
the identity operator on $\cL(A)$ and $T_B$ is the transpose operator on
$\cL(B)$.  We define $\PPT(A:B) = \{\rho\in\cD(A\ot B) : \rho^\Gamma
\in \cD(A:B)\}$.  This set is easier to work with because it has a
semidefinite-programming characterization.  Moreover,
it is straightforward to show that $\Sep(A:B)
\subset \PPT(A:B)$. However, in general this inclusion is strict, and
as the dimensions of $A,B$ grow large, $\PPT$ can be an arbitrarily
bad approximation for $\Sep$~\cite{BeigiS10}.

\medskip
\noindent
{\bf Partial trace.}  Another concept from probability theory that we
will need to generalize is the idea of a marginal distribution.  Say
we have a density matrix $\rho_{AB} \in \cD(A\ot B)$.  The subscript
emphasizes the systems which $\rho$ describes, which are analogous to
the random variables corresponding to a probability distribution.  To
obtain the state on only the $A$ system, we apply the {\em partial
  trace} operator $\tr_B := \id_A \ot \tr_B$ to $\rho_{AB}$.  The
action of the partial trace is often denoted by writing only the
subscripts, as in
\be \rho_A := \tr_B \rho_{AB}
\qquad\text{and}\qquad
\rho_B := \tr_A \rho_{AB}.\ee
(This notation generalizes; e.g. if $\rho\in\cD(A\ot B \ot C)$, then
$\rho_B = \tr_{AC} \rho_{ABC} = \tr_A \tr_C \rho_{ABC}$, etc.)
Concretely, $(\rho_A)_{i,i'} = \sum_j (\rho_{AB})_{(i,j), (i',j)}$ and $(\rho_B)_{j,j'} = \sum_i (\rho_{AB})_{(i,j), (i,j')}$.
We see that if $\rho$ is diagonal then this coincides with the idea of
a marginal distribution from classical probability theory.

\medskip
\noindent
{\bf Measurements.}  Although technically all of physics is described
by quantum mechanics, it is often convenient to make a distinction
between quantum information, which is often carried in very small
systems such as single atoms or single photons, and classical
information, which is carried in macroscopic systems, such as a bit in
a classical RAM.  The bridge from quantum state to probability
distribution is given by a {\em measurement} (also sometimes called a
POVM, for Positive-Operator-Valued Measure), which formally is a collection of matrices (POVM elements) $\cM =
(\cM_1,\ldots, \cM_k)$ satisfying $\cM_i \geq 0$ for each $i$ (meaning
each $\cM_i$ is positive semi-definite) and $\cM_1 + \cdots + \cM_k =
I$.  Performing the measurement $\cM$ on state $\rho$ yields outcome
$i$ with probability $\tr[\rho \cM_i]$.  Thus we can interpret $\cM$
as a linear map from $\cL(V)$ to $\bbR^k$, with the psd and
normalization conditions serving to guarantee that $\cM$ maps $\cD(V)$
to valid probability distributions.

\medskip
\noindent
{\bf Measurements on multipartite states.}  For our purposes, we will
consider a quantum state to be destroyed after it is measured.
However, if we have a quantum state on multiple systems, such as $A\ot
B$, and we measure only system $A$, then we will still have a quantum
state on system $B$.  In this case, the probability of obtaining
outcome $i$ is $\pr[i] = \tr[\cM_i\rho_A]$ and the residual state in this case
is
\be \frac{\tr_A [(\cM_i \ot I)\rho_{AB}]}{\pr[i]}.
\label{eq:residual}\ee
Since $\sum_i \cM_i= I $, we can verify that if we average over all
measurement outcomes, then system B is left in the state $\rho_B$,
independent of the choice of measurement.  This is an important
feature of quantum mechanics; despite the possibility of entanglement,
there is no way for Alice (who controls system $A$) to signal to Bob
(who controls system $B$) through her choice of measurement.

\medskip
\noindent
{\bf Restricted classes of measurements.}  Consider a bipartite system
$A\ot B$, with systems $A,B$ held by Alice and Bob respectively.
Performing a general measurement on $A\ot B$ may require that Alice
and Bob exchange quantum messages, so it is often more practical for
them to consider only measurements that they can perform using Local
Operations and Classical Communication (LOCC).  Although such
restricted measurements were initially introduced to model these
practical restrictions, they have since arisen in settings such as
\cite{BCY11, LW12} for completely different reasons.  The class LOCC
is difficult to work with and is cumbersome to even properly
define---see \cite{LOCC} for a discussion---so we will often work with
various restrictions or relaxations of it.  A restriction which is
interesting in its own right is the class $\ONELOCC$, which
corresponds to Alice performing a measurement locally and sending the
outcome to Bob.  We say that $\cM\in \ONELOCC$ if $\cM =
\{\cM_{i,j}\}$ with $\cM_{i,j} = X_i \ot Y_{i,j}$, each $X_i,
Y_{i,j}\geq 0$, $\sum_i X_i = I$ and for each $i$, $\sum_j Y_{i,j} =
I$. On the other hand, a useful relaxation is the set $\SEP$, for
which each $\cM_i$ should have the form $\cM_i = \sum_j X_{i,j} \ot
Y_{i,j}$ with each $X_{i,j}, Y_{i,j} \geq 0$.  An even further
relaxation is $\PPT$ for which we demand only that each $\cM_i^\Gamma
\geq 0$ (apart from the usual conditions that $\sum_i \cM_i = I$ and
each $\cM_i \geq 0$).  Finally we use $\ALL$ to denote the set of all
measurements.  Summarizing, we have
$$ \ONELOCC \subset \LOCC \subset \SEP
\subset \PPT \subset \ALL.$$
In each case, we consider measurements
with any finite number of outcomes, so these classes are technically
not compact.

\medskip
\noindent
{\bf Entanglement swapping.}
An important concept in our work (building on \cite{Piani09}) is that
of compatible pairs of families of measurements and states.  We say
that a POVM element $\cM_i$ is compatible with a family of
states $\famS$ if for each $n$ and each $\rho\in S^n$, applying $\cM_i$ to the first system
leaves a residual state (defined by \eq{residual}) that is in
$S^{n-1}$.  A family of measurements $\famM$ is compatible with $\famS$
if each POVM element of each measurement in $\famM$ is compatible with
$\famS$.  If $\famS = \Sep$, then $\ONELOCC, \LOCC,\SEP$ are all
compatible with $\famS$.  If $\famS=\PPT$ then the set of compatible
measurements includes $\PPT$.  However, it is easy to construct
examples of incompatible pairs.  Let $\ket{1},\ldots,\ket{d}$ be an
orthonormal basis of column vectors for $\bbC^d$
and define
$\ket\Psi = \frac{1}{d} \sum_{i,j \in [d]} \ket i \ot \ket j \ot \ket i \ot \ket j$.
Observe that $\Psi$ has entanglement between systems 1:3 and systems
2:4, but is product across the 13:24 cut.   Now consider a measurement
acting on systems 12.  One can
calculate that
\be \tr_{12} [(\cM_i \ot I)\proj\Psi] = \frac{\cM_i^T}{d}.\ee
Thus, if $\cM_i^T$ is proportional to an entangled state, then the
measurement can create entanglement on the previous unentangled states 3,4
that were not measured.  This phenomenon---in which we start with
$A_1:A_2$ and $B_1:B_2$ entanglement, measure $A_1B_1$ and end with
$A_2:B_2$ entanglement---is called entanglement
swapping~\cite{swapping} and is one of the main new difficulties
encountered in attempting to perform hypothesis testing with respect
to classes such as $\Sep$.

\medskip
\noindent
{\bf Entropy.}
The classical (Shannon) entropy of a distribution $p$ is given by
$H(p) = - \sum_i p_i \log(p_i)$.  The quantum analogue is called the
von Neumann entropy, and is given by $H(\rho) = -\tr [\rho \log \rho]$.
 Observe that $H(\rho)$ is
the Shannon entropy of the eigenvalues of $\rho$, and coincides with
the Shannon entropy when we consider probability distributions to be
diagonal density matrices.   If $\rho_{ABC}$ is a multipartite state,
then we let $H(A)_\rho := H(\rho_A), H(AB)_\rho = H(\rho_{AB}),$ etc.
When $\rho$ is understood, we may write simply $H(A), H(AB), \ldots$.
Analogous to the classical mutual information, conditional entropy,
etc. we can define
\ba H(A\mid B) &\defeq H(AB) - H(B) \twocol{\nn}\\
I(A:B) &\defeq H(A) + H(B) - H(AB) \twocol{\nn}\\
I(A:B\mid C) &\defeq H(AC) + H(BC) - H(ABC) - H(C), \twocol{\nn}
\ea
in each case with an implicit dependence on some state $\rho$.
Finally, the quantum relative entropy is $D(\rho \dmid \sigma) := \tr
[\rho(\log\rho - \log\sigma)]$.  Many of these quantities behave
similarly to their classical analogues, but a number of new subtleties
emerge; see Chapter 11 of \cite{Wilde-book} or Chapter 11 of
\cite{NC00} for more information.

\section*{Acknowledgments}
We are grateful to Keiji Matsumoto for helpful conversations about
hypothesis testing, and to the anonymous referee for going far beyond
their usual duties and correcting many mistakes, small and large, both
in our earlier versions and even in the classic text \cite{CT06}.  AWH
and FGSLB also thank the Mittag-Leffler Institute for their
hospitality while some of this work was done.  FGSLB was funded by
EPSRC. AWH was funded by NSF grants CCF-1111382, CCF-1452616,
CCF-1729369, PHY-1818914 and ARO contract W911NF-12-1-0486.  JRL was
supported by NSF grants CCF-1217256 and CCF-0905626.

\end{document}